\documentclass[aps,twocolumn,superscriptaddress,pra,10pt,showkeys]{revtex4-2}

\usepackage{algorithm2e}
\usepackage{amsfonts}
\usepackage{amsmath}
\usepackage{amssymb}
\usepackage{mathtools}
\usepackage{xfrac}

\usepackage{array}
\usepackage{graphicx}
\usepackage{diagbox}
\usepackage{multirow}

\usepackage{tikz}

\usetikzlibrary{shapes.geometric,positioning}
\tikzset{elli/.style={ellipse,draw}}
\tikzset{rect/.style={rectangle,draw}}
\usepackage[colorlinks = true, linkcolor = blue, urlcolor=blue, citecolor = red, anchorcolor = green,
]{hyperref}
\providecommand{\U}[1]{\protect\rule{.1in}{.1in}}
\newtheorem{theorem}{Theorem}

\newtheorem{definition}{Definition}

\newtheorem{lemma}{Lemma}

\newtheorem{problem}{Problem}

\newtheorem{remark}{Remark}

\newenvironment{proof}[1][Proof]{\noindent\textbf{#1.} }{\ \rule{0.5em}{0.5em}}

\newcommand{\bra}[1]{\langle#1|}
\newcommand{\ket}[1]{|#1\rangle}
\newcommand{\outerprod}[2]{\vert#1\rangle\!\langle#2 \vert}
\newcommand{\outerproj}[1]{\outerprod{#1}{#1}}

\newcommand{\Tr}{\operatorname{Tr}}
\newcommand{\minSymG}[2]{\min\limits_{\sigma \in \operatorname{Sym}_G} \frac{1}{2} \left\Vert #1 - #2 \right\Vert_1}
\newcommand{\maxSymG}[2]{\max\limits_{\sigma \in \operatorname{Sym}_G} F(#1 , #2)}
\newcommand{\qipeb}{QIP$_{\text{EB}}(2)$}

\newcommand{\td}[2]{\frac{1}{2} \left\Vert #1 - #2 \right\Vert_1}

\newcolumntype{P}[1]{>{\raggedright\arraybackslash}p{#1}}
\newcolumntype{L}{>{\centering\arraybackslash}m{2.6cm}}
\newcolumntype{M}[1]{>{\centering\arraybackslash}m{#1}}
\newcolumntype{C}[1]{>{\PreserveBackslash\centering}p{#1}}
\newcolumntype{R}[1]{>{\PreserveBackslash\raggedleft}m{#1}}

\allowdisplaybreaks

\begin{document}

\title[]{Quantum Computational Complexity and Symmetry}

\author{Soorya Rethinasamy}
\affiliation{School of Applied and Engineering Physics, Cornell University, Ithaca, New York 14850, USA}

\author{Margarite L. LaBorde}
\affiliation{Hearne Institute for Theoretical Physics, Department of Physics and Astronomy, and Center for Computation and Technology, Louisiana State University, Baton Rouge, Louisiana 70803, USA}

\author{Mark M. Wilde}
\email{wilde@cornell.edu}
\affiliation{School of Electrical and Computer Engineering, Cornell University, Ithaca, New York 14850, USA}

\begin{abstract}
Testing the symmetries of quantum states and channels provides a way to assess their usefulness for different physical, computational, and communication tasks. Here, we establish several complexity-theoretic results that classify the difficulty of symmetry-testing problems involving a unitary representation of a group and a state or a channel that is being tested.
In particular, we prove that various such symmetry-testing problems are complete for BQP (bounded-error quantum polynomial time), QMA (quantum Merlin-Arthur), QSZK (quantum statistical zero-knowledge), QIP(2) (two-message quantum interactive proofs), \qipeb (two-message quantum interactive proofs restricted to entanglement-breaking provers), and QIP (quantum interactive proofs), thus spanning the prominent classes of the quantum interactive proof hierarchy and forging a non-trivial connection between symmetry and quantum computational complexity.
Finally, we prove the inclusion of two Hamiltonian symmetry-testing problems in QMA and QAM, while leaving it as an intriguing open question to determine whether these problems are complete for these classes. 
% \mmw{Feedback from Ravi Rau:
% Thanks very much for your submission. I was not aware till I saw your message that Sai has sought these by end-August! The paper is far too technical for me to absorb but I did note immediately that the Table on p.5 is a useful capsuling of many results. It may be useful for you to spell out all the abbreviations such as BQP, QAM, and QMA that are very familiar to you but will not be to others. Not just in the Abstract but even in the first sections, I was still at a loss till Sec. III when you finally indicate the abbreviations in brackets. Since Abstracts are sometimes presented separate from the full article, it would be best to have this already in the Abstract.}
\end{abstract}


\date{\today}
\maketitle
\tableofcontents

\section{Introduction}
Computational complexity theory was born from the Church--Turing thesis, which informally states that ``any real-world computation can be translated into an equivalent computation involving a Turing machine." Turing machines soon became the standard method to measure the difficulty or computational complexity of a problem. However, they are classical machines and, as such, are limited by classical physics. With the advent of quantum mechanics as a prevalent model for explaining the universe, the question of whether Turing machines capture all that is computable arose naturally. To address this issue, a new computational model, called the quantum Turing machine, was developed to incorporate quantum mechanics \cite{benioff1980computer,benioff1982quantum,DP85}. This led to the development of quantum computational complexity theory, which classifies computational problems by means of quantum Turing machines \cite{BZ97}. Later on, it was shown that the quantum Turing machine and quantum circuit computational models are equivalent \cite{Chi93} (see also \cite{MW19}), and the bulk of the larger research community's focus has been on the quantum circuit model since then.

Quantum computational complexity theory provides an important contextualization for quantum algorithms; namely, determining the complexity class of a computational problem allows for a meaningful discussion of how difficult it is for a quantum computer to solve the problem as its scale increases \cite{watrous2009complexity,VW15}. This classification allows us to clearly separate problems into broad categories, like those that are efficiently solvable on a quantum computer, those that are efficiently solvable on a classical computer, etc. The various classical and quantum complexity classes can be arranged into a hierarchy, which gives insight into their relative difficulty. For instance, an NP-hard problem is believed to be inefficiently solvable, but a PSPACE-hard problem is believed to be significantly more so. The complexity-theoretic perspective becomes relevant whenever attempting to solve computationally difficult tasks, such as in machine-learning applications. 

In this paper, we connect several major quantum complexity classes to a hierarchy of symmetry-testing tasks. Much interest has arisen in using symmetry in quantum information science for various applications, linked with its fundamental role in physics \cite{FR96,Gross96}. To give a few examples, symmetries can be used to test for separability of pure states, as in \cite{BBD+97,brennen2003entanglement,HM10,GHMW15,BGCC21,bradshaw2022cycle}. In both classical and quantum machine learning tasks, symmetries can vastly improve performance by limiting the relevant search space \cite{LSS+22,Skolik2022graph,Meyer2022exploiting}, and machine learning can help to identify hidden symmetries manifested only in certain coordinate systems \cite{Liu2022symmetries}. Symmetries in POVMs have been utilized in state discrimination \cite{KGDdS15} and estimation \cite{CD04} applications. Additionally, recent work has been conducted to test symmetries of states, channels, measurements \cite{LRW22,bandyopadhyay2023efficient}, Hamiltonians \cite{laborde2022hamiltonian}, and Lindbladians \cite{bandyopadhyay2023efficient} using quantum algorithms. It is on these last three works that we expand, tying symmetry-testing language to a complexity-theoretic structure.

In more detail, we provide a number of complexity-theoretic characterizations for symmetry-testing tasks. See Section~\ref{sec:classreview} for a review of various complexity classes, as well as \cite{watrous2009complexity,VW15} for much more detail. Table~\ref{tab:SymTestComplexity} and Figure~\ref{fig:ClassComplexityDiagram} provide all the results and details of our paper at a glance, and Figure~\ref{fig:ClassComplexityDiagram} places some of them in a containment diagram for ease of access. These results span a symmetry-testing problem complete for BQP to another complete for QIP.  More specifically, for the classes BQP, QMA, QIP(2), \qipeb, QSZK, and QIP, we provide symmetry-testing problems that are complete for each class, meaning any promise problem in those classes can be mapped in polynomial time to a respective symmetry-testing problem for some group~$G$. In particular, we prove complexity-theoretic results for the various symmetry definitions  given in \cite{LRW22}---for a given input state, testing for $G$-symmetry with respect to Hilbert--Schmidt norm is BQP-Complete, testing for $G$-Bose symmetry is BQP-Complete, testing for $G$-symmetry with respect to trace norm is QSZK-Complete, and testing for $G$-Bose symmetric extendibility is QIP(2)-Complete. The last aforementioned contribution, on QIP(2)-completeness, addresses a problem that has been open in quantum computational complexity theory for quite some time \cite{JUW09} --- namely, to identify a non-trivial problem complete for the class other than the known Close Image problem. In addition, we establish a symmetry-testing problem that is complete for \qipeb, and we show that testing whether there exists an input to a channel such that its output is $G$-Bose symmetric or $G$-Bose symmetric extendible is QMA-Complete and QIP-Complete, respectively. Furthermore, we show that two different versions of Hamiltonian symmetry-testing problems are in~QMA and QAM.

We accomplish these findings in the following ways. To establish containment of a given symmetry-testing problem in a given class, we either employ the algorithms put forward in \cite{LRW22,bandyopadhyay2023efficient}, which involve a verifier acting with or without a prover, or we propose new algorithms. To establish hardness of a given symmetry-testing problem, the key concept is to embed the circuits involved in a general computation either into the preparation of the state or channel being tested or into a unitary representation of a group. This approach to proving hardness results is common in the literature on quantum interactive proofs \cite{watrous2002qszk,RW05,GHMW15,VW15,RASW23}, as well as in a recent paper \cite{laborde2022hamiltonian} establishing DQC1-hardness of a Hamiltonian symmetry-testing problem.

The outlook and findings of our paper complement those of \cite{HMW13,GHMW15} and \cite{kitaev2000parallelization,W02QIP,RW05,HMW13,RASW23}, which connected entanglement theory and distinguishability, respectively, to quantum computational complexity theory. Here, in the same spirit, we connect symmetry to quantum computational complexity theory. To recall the prior developments, the authors of \cite{HMW13,GHMW15}  showed that, for most quantum complexity classes that involve interaction with a computationally-unbounded prover, including BQP, QMA,  QSZK, QMA(2), and QIP, there exists a corresponding natural separability testing problem that is complete for the class. These findings are summarized in \cite[Figure~1]{GHMW15}. Similarly, in \cite{kitaev2000parallelization,W02QIP,RW05,W09zkqa,HMW13,RASW23}, these authors showed that there exist natural distinguishability problems that are complete for the same aforementioned classes, as well as QIP(2), as summarized in \cite[Figure~17]{RASW23}. Here, we show that there are multiple symmetry-testing problems that are complete for a range of quantum complexity classes.

The rest of our paper is structured as follows. Section~\ref{sec:prelim} begins with some preliminaries that are helpful for understanding the remainder of the work. In Section~\ref{sec:classreview}, we review the class definitions for the relevant promise problems. Section~\ref{sec:results} contains a variety of complexity-theoretic results for different symmetry tests. In Section~\ref{sec:GBS-BQP}, we prove that testing $G$-Bose symmetry of a state $\rho$ is BQP-Complete. In Section~\ref{sec:Sym-HS-BQP}, we show that testing $G$-symmetry of a state, according to a Hilbert--Schmidt norm measure, is BQP-Complete. In Section~\ref{sec:GBS-QMA}, we go on to show that testing if there exists an input to a channel such that its output is $G$-Bose symmetric is QMA-Complete. Following this, Sections~\ref{sec:GS-TD-QSZK} and \ref{sec:GS-Fid-QSZK} show that the problems of testing $G$-symmetry of a state, according to a trace-norm measure and fidelity, respectively, are QSZK-Complete. In Section~\ref{sec:GBSE-QIP2}, testing if a state $\rho$ is $G$-Bose symmetric extendible is shown to be QIP(2)-Complete; while in Section~\ref{sec:GBSE-QIP}, we show that the channel version of the problem is QIP-Complete. In Section~\ref{sec:GBS-QIPEB}, we show that the problem of testing whether a state has a separable, $G$-Bose symmetric extension is \qipeb-Complete. Finally, in Sections~\ref{sec:HS-QMA} and \ref{sec:HS-QAM}, the problem of estimating the maximal and average spectral norm of the commutator between a Hamiltonian and a group representation $\{U(g)\}_{g\in G}$ is shown to be in QMA and QAM, respectively. We conclude in Section~\ref{sec:conclusion} with a summary and some suggestions for future directions.

\section{Notions of Symmetry}
\label{sec:prelim}

In this section, we review the notions of symmetry presented in \cite{MS13,MS14,LRW22} with respect to some finite group $G$ and a corresponding unitary representation $\{U(g)\}_{g\in G}$. We also introduce two notions of symmetry not previously considered in \cite{MS13,MS14,LRW22}, which are related to and generalize the symmetry considered recently in \cite{philip2023quantum}. Our task in this work will be to contextualize these symmetry definitions in a quantum complexity-theoretic framework. 

\subsection{Review of Existing Notions of Symmetries}

\begin{definition}[$G$-symmetric]
\label{def:G-symmetry}
Let $G$ be a group with projective unitary representation $\{U_{S}(g)\}_{g\in G}$, and let $\rho_{S}$ be a state of system $S$. The state $\rho_{S}$ is symmetric with respect to $G$ \cite{MS13,MS14}\ if
\begin{equation}
\rho_{S}=U_{S}(g)\rho_{S}U_{S}(g)^{\dag}\quad\forall g\in G.
\end{equation}
\end{definition}

$G$-symmetry is the usual notion of symmetry considered in most physical contexts. For example, in \cite{RBN+22, LSS+22}, the authors use $G$-symmetric states in various quantum machine learning applications, primarily in classification algorithms where the labeling of the state should remain invariant. Additionally, testing the incoherence of a state in the vein of \cite{SAP17,BCP14} is a special case of a $G$-symmetry test where the group is the cyclic group of order $|G|$.

Expanding upon this definition, we recall the definition of $G$-Bose symmetry, a stronger notion of symmetry. $G$-Bose symmetry implies $G$-symmetry, though the reverse implication is not true in general. $G$-Bose symmetry checks if a state belongs to the symmetric subspace induced by the group representation. This more mathematical notion of symmetry has proven useful in deriving important results, such as the quantum de Finetti theorem \cite{harrow2013church}. As a practical application, a circuit construction for projecting onto the symmetric subspace corresponding to the standard symmetric group \cite{BBD+97} has been used in a number of quantum computational tests of entanglement \cite{HMW13,GHMW15,LRW22,bradshaw2022cycle}. We give the definition of $G$-Bose symmetry below.

\begin{definition}[$G$-Bose-symmetric]
\label{ex:usual-Bose-symmetry}
Let $G$ be a group with unitary representation $\{U_{S}(g)\}_{g\in G}$.
A state $\rho_{S}$ is Bose-symmetric with respect to $G$ if
\begin{equation}
\rho_{S}=U_{S}(g)\rho_{S}\quad\forall g\in G.
\label{eq:def-G-Bose-sym}
\end{equation}
The condition in \eqref{eq:def-G-Bose-sym} is equivalent to the condition
\begin{equation}
\rho_{S}=\Pi_{S}^{G}\rho_{S}\Pi_{S}^{G},
\end{equation}
where the projector $\Pi_{S}^{G}$ is defined as
\begin{equation}
\label{eq:group_proj_GBS}
\Pi_{S}^{G}\coloneqq \frac{1}{\left\vert G\right\vert }\sum_{g\in G}U_{S}(g).
\end{equation}
\end{definition}

We note here that the notion of $G$-Bose-symmetry can be generalized to compact Lie groups. In this case, one requires an invariant measure, which exists for all such groups. An important example in quantum information is the unitary group equipped with the Haar measure. See \cite[Proposition~2]{harrow2013church} for mathematical details of how the projector $\Pi_S^G$ is defined in this case. Throughout our paper, however, we focus exclusively on finite groups.

Both of the aforementioned symmetry notions ($G$-symmetry and $G$-Bose-symmetry) can be expanded to scenarios in which the tester has limited access to the state of interest. For example, one can test whether, given a part of a state, there exists an extension that is symmetric. These notions lead to further, pertinent symmetry tests. For instance, when the group in question is the permutation group, $G$-Bose extendibility is relevant for detecting entanglement \cite{NOP09} and efficiently bounding quantum discord \cite{Pia16}. Similarly, $G$-symmetric extendible states have been studied in the context of entanglement distillability \cite{Now16} and $k$-extendibility \cite{W89a,DPS02,DPS04,BFC12,KDWW19}.

\begin{definition}[$G$-symmetric extendible]
\label{def:g-sym-ext}
Let $G$ be a group  with unitary representation $\{U_{RS}(g)\}_{g\in G}$.
A state~$\rho_{S}$ is $G$-symmetric extendible if there exists a state $\omega_{RS}$ such that
\begin{enumerate}
\item the state $\omega_{RS}$ is an extension of $\rho_{S}$, i.e.,
\begin{equation}
\operatorname{Tr}_{R}[\omega_{RS}]=\rho_{S}, \label{eq:G-ext-1}
\end{equation}
\item the state $\omega_{RS}$ is $G$-symmetric, in the sense that
\begin{equation}
\omega_{RS}=U_{RS}(g)\omega_{RS}U_{RS}(g)^{\dag}\qquad\forall g\in G.
\label{eq:G-ext-2}
\end{equation}
\end{enumerate}
\end{definition}

\begin{definition}[$G$-Bose symmetric extendible]
\label{def:g-bose-sym-ext}
A state $\rho_{S}$ is $G$-Bose symmetric extendible (G-BSE) if there exists a state $\omega_{RS}$ such that
\begin{enumerate}
\item the state $\omega_{RS}$ is an extension of $\rho_{S}$, i.e.,
\begin{equation}
\operatorname{Tr}_{R}[\omega_{RS}]=\rho_{S},
\end{equation}
\item the state $\omega_{RS}$ is Bose symmetric, i.e., satisfies
\begin{equation}
\omega_{RS}=\Pi_{RS}^{G}\omega_{RS}\Pi_{RS}^{G},
\label{eq:bose-G-sym-ext-cond}
\end{equation}
where
\begin{equation}
    \Pi_{RS}^{G}\coloneqq\frac{1}{\left\vert G\right\vert }\sum_{g\in G}U_{RS}(g).
    \label{eq:Pi-proj-RS-def}
\end{equation}
\end{enumerate}
\end{definition}

In \cite[Section~5]{LRW22}, we discussed how each of the above symmetry definitions can be expressed as semi-definite programs (SDPs), thus implying that it is efficient to test for them when the state and unitary representations are given as matrices. That is, the runtime of these SDP algorithms is polynomial in the dimension of the states and unitaries. In this paper, we consider the complexity of testing these symmetries when circuit descriptions are given for preparing these states and executing these unitary representations on quantum computers.  

\subsection{Separably Extendible Symmetries}

Let us finally introduce two other notions of symmetry, one of which represents a generalization of a symmetry recently considered in \cite{philip2023quantum}. Before doing so, let us first recall that a bipartite state $\rho_{AB}$ is separable with respect to the partition $A,B$, denoted as $\rho_{AB} \in \operatorname{SEP}\!\left(A\!:\!B\right)$, if it can be written in the following form \cite{Werner89}:
\begin{equation}
    \rho_{AB} = \sum_{x\in \mathcal{X}} p(x) \psi^x_A \otimes \phi^x_B,
\end{equation}
where $\mathcal{X}$ is a finite alphabet, $\{p(x)\}_{x\in\mathcal{X}}$ is a probability distribution, and $\{\psi^x_A\}_{x\in \mathcal{X}}$ and $\{\phi^x_B\}_{x\in \mathcal{X}}$ are sets of pure states. States that cannot be written in this form are entangled.

\begin{definition}[$G$-symmetric separably extendible]
\label{def:g-sym-sep-ext}
Let $G$ be a group with projective unitary representation $\{U_{RS}(g)\}_{g\in G}$, and let $\rho_{S}$ be a state.
The state $\rho_{S}$ is $G$-symmetric separably extendible if there exists a state $\omega_{RS} $ such that
\begin{enumerate}
\item the state $\omega_{RS}$ is a separable extension of $\rho_{S}$, i.e.,
\begin{align}
\operatorname{Tr}_{R}[\omega_{RS}] & =\rho_{S}, \label{eq:G-sep-ext-1}\\
\omega_{RS} & \in \operatorname{SEP}(R\!:\!S),
\end{align}
\item the state $\omega_{RS}$ is $G$-symmetric, in the sense that
\begin{equation}
\omega_{RS}=U_{RS}(g)\omega_{RS}U_{RS}(g)^{\dag}\qquad\forall g\in G.
\label{eq:G-sep-ext-2}
\end{equation}
\end{enumerate}
\end{definition}

\begin{definition}[$G$-Bose-symm.~separably extendible]
\label{def:g-bose-sym-sep-ext}
Let $G$ be a group with unitary representation $\{U_{RS}(g)\}_{g\in G}$, and let $\rho_{S}$ be a state.
The state $\rho_{S}$ is $G$-Bose-symmetric separably extendible if there exists a state $\omega_{RS} $ such that
\begin{enumerate}
\item the state $\omega_{RS}$ is a separable extension of $\rho_{S}$, i.e.,
\begin{align}
\operatorname{Tr}_{R}[\omega_{RS}] & =\rho_{S}, \label{eq:G-bose-sep-ext-1}\\
\omega_{RS} & \in \operatorname{SEP}(R\!:\!S),
\end{align}
\item the state $\omega_{RS}$ is Bose symmetric, i.e., satisfies
\begin{equation}
\omega_{RS}=\Pi_{RS}^{G}\omega_{RS}\Pi_{RS}^{G},
\label{eq:bose-G-bose-sym-sep-ext-cond}
\end{equation}
where $\Pi_{RS}^{G}$ is defined in \eqref{eq:Pi-proj-RS-def}.
\end{enumerate}
\end{definition}

By comparing Definitions~\ref{def:g-sym-ext} and \ref{def:g-bose-sym-ext} with Definitions~\ref{def:g-sym-sep-ext} and \ref{def:g-bose-sym-sep-ext}, respectively, we see that the main additional constraint in the latter definitions is that the extension is required to be a separable state. As such, when the state and unitary representations are given as matrices, this additional constraint makes the search for an extension more computationally difficult than those needed for Definitions~\ref{def:g-sym-ext} and \ref{def:g-bose-sym-ext}, because optimizing over the set of separable states is computationally difficult \cite{Gur03,Gha10} and it is not possible to perform this search by means of SDPs \cite{fawzi2021set}. Here, we consider the complexity of testing the symmetry in Definition~\ref{def:g-bose-sym-sep-ext} when the state and unitary representations are given as circuit descriptions. 

Let us comment briefly on the connection between Definition~\ref{def:g-bose-sym-sep-ext} and the symmetry considered in \cite{philip2023quantum}. In \cite{philip2023quantum}, the goal was to test whether a given bipartite state $\rho_{AB}$ is separable. It was shown that one can equivalently do so by testing whether there exists a separable extension $\rho_{A'AB} \in \operatorname{SEP} \!\left(A'\!:\!AB\right)$ of $\rho_{AB}$ that is Bose symmetric with respect to the unitary representation $\{I_{A'A}, F_{A'A}\}$ of the symmetric group of order two, where $F_{AA'}$ is the unitary swap operator. More concretely, the test checks whether there exists $\rho_{A'AB} \in \operatorname{SEP}(A'\!:\!AB)$ such that
\begin{align}
    \Tr_{A'}[\rho_{A'AB}] &  = \rho_{AB},\\
    \rho_{A'AB} & = \Pi_{A'A} \rho_{A'AB} \Pi_{A'A},
\end{align}
where $\Pi_{A'A} \coloneqq (I_{AA'} + F_{AA'})/2$. As such, this represents a non-trivial example of the symmetry presented in Definition~\ref{def:g-bose-sym-sep-ext}.

Quantum algorithms to test the symmetries in Definitions~\ref{def:G-symmetry}--\ref{def:g-bose-sym-ext} were given in \cite{LRW22}, and we provide a new quantum algorithm to test for the symmetry in Definition~\ref{def:g-bose-sym-sep-ext}. More specifically, we invite the reader to examine the following parts of \cite{LRW22} to better understand the connections between the definitions of symmetry and quantum algorithms for testing them:
\begin{enumerate}
    \item Algorithm~1 therein, depicted in Figure~1 therein, tests for $G$-Bose symmetry.
    \item Algorithm~2 therein, depicted in Figure~4 therein, tests for $G$-symmetry.
    \item Algorithm~3 therein, depicted in Figure~6 therein, tests for $G$-Bose symmetric extendibility.
    \item Algorithm~4 therein, depicted in Figure~8 therein, tests for $G$-symmetric extendibility.
\end{enumerate}
In the following sections, we also discuss the complexity of several corresponding symmetry-testing problems.

\renewcommand{\arraystretch}{2.05}
\begin{table*}
\begin{center}
    \begin{tabular}{|M{3.2cm}|M{6.5cm}|M{4.2cm}|M{2.8cm}|}
    \hline 
        \multirow[c]{2}{=}{\centering \textbf{Problem}} & 
        \multirow[c]{2}{=}{\centering \textbf{Description}} & 
        \multirow[c]{2}{=}{\centering \textbf{Mathematical Description}} & 
        \multirow[c]{2}{=}{\centering \textbf{Complexity}}\\
        & & & \\
        \hline
        \multirow[c]{2}{=}{\centering \hyperref[sec:GBS-BQP]{State $G$-Bose Symmetry}} & \multirow[c]{2}{=}{\centering Decide if a state  is $G$-Bose symmetric} & \multirow[c]{2}{=}{\centering $\Tr\!\left[\Pi^G \rho \right]$} & \multirow[c]{2}{=}{\centering BQP-Complete}\\
        & & & \\
        \hline
        \multirow[c]{2}{=}{\centering \hyperref[sec:Sym-HS-BQP]{State $G$-Symmetry in Hilbert--Schmidt Norm}} & \multirow[c]{2}{=}{\centering Decide if a state is $G$-symmetric in Hilbert--Schmidt norm} & \multirow[c]{2}{=}[-0.1cm]{\centering $\frac{1}{\vert G \vert} \sum\limits_{g \in G} \left\Vert [U(g), \rho] \right\Vert^2_2$ } & \multirow[c]{2}{=}{\centering BQP-Complete}\\
        & & & \\
        \hline
        \multirow[c]{2}{=}{\centering \hyperref[sec:GBS-QMA]{Channel $G$-Bose Symmetry}} & \multirow[c]{2}{=}{\centering Decide if the output of a channel with \\ optimized input is $G$-Bose symmetric} & \multirow[c]{2}{=}[-0.1cm]{\centering $\max\limits_{\rho \in \mathcal{D}(\mathcal{H})}\  \Tr\! \left[ \Pi^G \mathcal{N}(\rho) \right]$} & \multirow[c]{2}{=}{\centering QMA-Complete}\\
        & & & \\
        \hline
        \multirow[c]{2}{=}{\centering \hyperref[sec:GS-TD-QSZK]{State $G$-Symmetric Trace Distance}} & \multirow[c]{2}{=}{\centering Decide if a state is $G$-symmetric\\in trace norm} & \multirow[c]{2}{=}[-0.1cm]{\centering $\min\limits_{\sigma \in \operatorname{Sym}_G} \frac{1}{2} \left\Vert \rho - \sigma \right\Vert_1$} & \multirow[c]{2}{=}{\centering QSZK-Complete}\\
        & & & \\
        \hline
        \multirow[c]{2}{=}{\centering \hyperref[sec:GS-Fid-QSZK]{State $G$-Symmetric Fidelity}} & \multirow[c]{2}{=}{\centering Decide if a state is $G$-symmetric\\using fidelity} & \multirow[c]{2}{=}[-0.1cm]{\centering $\max\limits_{\sigma \in \operatorname{Sym}_G} F(\rho, \sigma)$} & \multirow[c]{2}{=}{\centering QSZK-Complete}\\
        & & & \\
        \hline
        \multirow[c]{2}{=}{\centering \hyperref[sec:GBSE-QIP2]{State $G$-Bose Symmetric Extendibility}} & \multirow[c]{2}{=}{\centering Decide if a state is $G$-Bose \\  symmetric extendible} & \multirow[c]{2}{=}{\centering $\max\limits_{\sigma \in \operatorname{BSE}_G} F(\rho, \sigma)$} & \multirow[c]{2}{=}{\centering QIP(2)-Complete}\\
        & & & \\
        \hline
        \multirow[c]{2}{=}{\centering \hyperref[sec:GBS-QIPEB]{State Separable Extension $G$-Bose Symmetry}} & \multirow[c]{2}{=}{\centering Decide if a state has a separable $G$-Bose symmetric extension} & \multirow[c]{2}{=}[-0.1cm]{\centering $\max\limits_{\substack{\omega_{RS} \in \operatorname{SEP}(R:S), \\ \Tr_R[\omega_{RS}] = \rho_S}} \Tr\!\left[\Pi^G_{RS} \omega_{RS}\right]$} & \multirow[c]{2}{=}{\centering \qipeb-Complete}\\
        & & & \\
        \hline
        \multirow[c]{2}{=}{\centering \hyperref[sec:GBSE-QIP]{Channel $G$-Bose Symmetric Extendibility}} & \multirow[c]{2}{=}{\centering Decide if the output of a channel with \\ optimized input is $G$-Bose symmetric \\ extendible} & \multirow[c]{2}{=}[-0.1cm]{\centering $\max\limits_{\substack{\rho \in \mathcal{D}(\mathcal{H}), \\\sigma \in \operatorname{BSE}_G}} F(\mathcal{N}(\rho), \sigma)$} & \multirow[c]{2}{=}{\centering QIP-Complete}\\
        & & & \\
        \hline
        \multirow[c]{2}{=}{\centering \hyperref[sec:HS-QMA]{ Hamiltonian Symmetry Maximal Spectral Norm}} & \multirow[c]{2}{=}{\centering Decide if a Hamiltonian is symmetric in maximum spectral norm} & \multirow[c]{2}{=}[-0.1cm]{ \centering $\max\limits_{g \in G} \left\Vert [U(g), e^{-iHt}] \right\Vert^2_\infty$} & \multirow[c]{2}{=}{\centering QMA}\\
        & & & \\
        \hline
        \multirow[c]{2}{=}{\centering \hyperref[sec:HS-QAM]{Hamiltonian Symmetry Average Spectral Norm}} & \multirow[c]{2}{=}{\centering Decide if a Hamiltonian is symmetric in average spectral norm} & \multirow[c]{2}{=}[-0.1cm]{ \centering $\frac{1}{|G|} \sum\limits_{g \in G} \left\Vert [U(g), e^{-iHt}] \right\Vert^2_\infty$} & \multirow[c]{2}{=}{\centering QAM}\\
        & & & \\
        \hline
    \end{tabular}
\caption{List of all complexity-theoretic results from this work. In this table, $\rho$ and $\sigma$ are mixed states, and $\mathcal{N}$ is a quantum channel. $\{U(g)\}_{g \in G}$ is a unitary representation of the group $G$. The shorthands  $\operatorname{Sym}_G$, $\operatorname{B-Sym}_G$, and $\operatorname{BSE}_G$ denote the sets of  $G$-symmetric, $G$-Bose symmetric, and $G$-Bose-symmetric extendible states for the chosen unitary representation of $G$, respectively (from Definitions~\ref{def:G-symmetry}, \ref{ex:usual-Bose-symmetry},  and \ref{def:g-bose-sym-ext}). $\operatorname{SEP}(R\!:\!S)$ denotes the set of separable states across the $R, S$ partition. Finally, $\Pi^G$ denotes the projector onto the symmetric subspace of the group representation, as defined in~\eqref{eq:group_proj_GBS}. }
\label{tab:SymTestComplexity}
\end{center}
\end{table*}

\section{Review of Quantum Computational Complexity Theory and Classes} \label{sec:classreview}

\renewcommand{\arraystretch}{1.2}
\begin{figure}
\begin{tikzpicture}
[
    level 1/.style = {rect, level distance=2cm, sibling distance = 5cm},
    level 2/.style = {rect, sibling distance = 5cm, level distance=5cm },
    every node/.style={rectangle,draw,minimum width=3cm}
]
    \node at (4.5,-2.0) (qipeb) {\begin{tabular}{c} \textbf{\qipeb-Complete} \rule{0pt}{5pt} \\ \hline \rule{0pt}{11pt}$\max\limits_{\substack{\omega_{RS} \in \operatorname{SEP}(R:S), \\ \Tr_R[\omega_{RS}] = \rho_S}} \Tr\!\left[\Pi^G_{RS} \omega_{RS}\right]$ \end{tabular}};
    
    \node at (0.0, 0.0) (qip) {\begin{tabular}{c} \textbf{QIP-Complete} \rule{0pt}{5pt} \\ \hline \rule{0pt}{11pt} $\max\limits_{\substack{\rho \in \mathcal{D}(\mathcal{H}), \\\sigma \in \operatorname{BSE}_G}} F(\mathcal{N}(\rho), \sigma)$ \end{tabular}};
    
    \node (qip2) at (0.0, -2.0) {\begin{tabular}{c} \textbf{QIP(2)-Complete} \rule{0pt}{5pt} \\ \hline \rule{0pt}{10pt}$\max\limits_{\sigma \in \operatorname{BSE}_G} F(\rho, \sigma)$ \end{tabular}};
    
    \node [below of=qip2, yshift=-1.25cm](qszk) {\begin{tabular}{c} \textbf{QSZK-Complete} \rule{0pt}{5pt} \\ \hline \rule{0pt}{11pt}$\min\limits_{\sigma \in \operatorname{Sym}_G} \frac{1}{2} \left\Vert \rho - \sigma \right\Vert_1$\\$\max\limits_{\sigma \in \operatorname{Sym}_G} F(\rho, \sigma)$ \end{tabular}};
    
    \node [below of=qipeb, yshift=-1.25cm] (qma) {\begin{tabular}{c} \textbf{QMA-Complete} \\ \hline \rule{0pt}{11pt}$\max\limits_{\rho\in \mathcal{D}(\mathcal{H})}\  \Tr\!\left[ \Pi^G \mathcal{N}(\rho) \right]$ \end{tabular}};
    
    \node at (2.25, -6.75) (bqp) {\begin{tabular}{c} \textbf{BQP-Complete} \\ \hline \rule{0pt}{11pt}$\Tr\!\left[\Pi^G \rho \right]$ \\ \rule{0pt}{8pt}$\frac{1}{\vert G \vert} \sum\limits_{g \in G} \left\Vert [U(g), \rho] \right\Vert^2_2$ \end{tabular}};

\draw (bqp) -- (qszk);
\draw (qip2) -- (qszk);
\draw (qip2) -- (qip);
\draw (qma) -- (qip2);
\draw (bqp) -- (qma);
\draw (qipeb) -- (qszk);
\draw (qipeb) -- (qma);
\end{tikzpicture}
\caption{List of complete symmetry-testing problems and the corresponding quantum complexity class. The notations are the same as described in Table~\ref{tab:SymTestComplexity}. The cells are organized such that if a cell is connected to a cell above it, the complexity class for the lower cell is a subset of that for the higher cell. For example, QMA is a subset of QIP(2).}
\label{fig:ClassComplexityDiagram}
\end{figure}
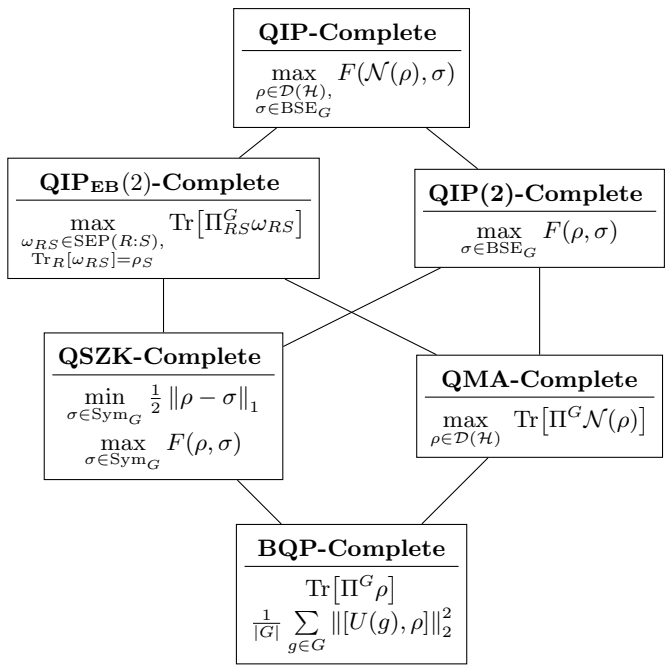

Before we delve into the different results of this work, we present a short review of quantum computational complexity theory. As mentioned previously, comprehensive reviews can be found in \cite{watrous2009complexity, VW15}.

Complexity theory is the branch of study that broadly seeks to classify the difficulty of computational problems \cite{Goldreich2008,arora2009computational}. This field of study has led to answers to several questions of the form: ``Is Problem $A$ harder to solve than Problem $B$?" or ``Given a potential solution to Problem $A$, is it efficient to verify its authenticity?" Broadly speaking, computational problems are placed into computational complexity classes. All problems within a class can be thought of as having similar complexity or difficulty. Examples of such complexity classes are P (the class of problems efficiently solvable on a classical computer), NP (the class of problems whose solutions can be verified efficiently), etc. Quantum complexity theory is a generalization of classical complexity theory where the input resources and the underlying computation are allowed to be quantum mechanical. An important example is BQP, the class of problems efficiently solvable using quantum computers. 

In anticipation of the results of this paper, we define two more concepts that will be essential -- `hardness' and `completeness'. A problem is said to be hard for a computational class if it is at least as hard to solve as the hardest problem of the class. 
A problem is complete for a class if it is hard for the class, and additionally, belongs to the class. A property of complete problems is that every other problem in the class can be efficiently mapped to a complete problem. In other words, the ability to solve a complete problem for a class can be efficiently repurposed to solve any other problem in that class. Therefore, a complete problem indeed completely characterizes the difficulty of the class. 

Two different methods exist to show that a problem is hard for a given class. First, we pick another problem that is known to be complete for the class and efficiently map that problem to the problem of interest. Another method is to take the definition of the class itself and show that an arbitrary problem in the class can be efficiently mapped to the problem of interest. 

Another important concept needed to fully specify the complexity of a problem is a polynomial-time generated family of circuits. Given a classical description/encoding of a quantum circuit, $x \in S$, where $S \subseteq \{0, 1\}^*$ is a set of binary strings, the set of quantum circuits $\left\{Q_x \mid x \in S\right\}$ is said to be polynomial-time generated if there exists a Turing machine that takes in as input the string $x$ and outputs an encoding of the quantum circuit $Q_x$ in polynomial time. This particular definition allows us to limit the power of the computational model to circuits that are ``polynomially complex" by limiting the process by which such circuits are created.

Lastly, we define promise problems. A promise problem can be thought of as a yes-no rewriting of a general decision problem. More concretely, a promise problem is a pair $L = (L_{\operatorname{yes}}, L_{\operatorname{no}})$, where $L_{\operatorname{yes}}, L_{\operatorname{no}}$ are subsets of all possible inputs such that $L_{\operatorname{yes}} \cap  L_{\operatorname{no}}  = \emptyset$. The inputs of the two subsets are called yes-instances and no-instances, respectively. An algorithm is said to ``decide" a promise problem if, given an input from $L_{\operatorname{yes}} \cup L_{\operatorname{no}}$, it can determine to which subset the input belongs.

Throughout this work, we will make reference to various established complexity classes via their archetypal, complete promise problems. In view of both conciseness and comprehensiveness, those relevant promise problems are presented here. 

\subsection{BQP}

\label{subsec:BQP}

The class of bounded-error quantum polynomial time (BQP) promise problems is often referred to as the class of problems efficiently solvable on a quantum computer \cite[Chapter~4]{nielsenchuang}. The classical analog of BQP is the class of bounded-error probabilistic polynomial time (BPP) problems, which is the class of problems efficiently solvable on a classical computer with access to random bits. A promise problem is a member of BQP if there exists an efficient quantum algorithm solving it in polynomial time with a success probability of at least $2/3$.

Here, we give the definition of BQP for convenience, wherein we restrict the circuits considered to be unitary circuits; it is known that the computational power of BQP does not change under this restriction. Let $L = (L_{\operatorname{yes}}, L_{\operatorname{no}})$ be a promise problem, $\alpha, \beta : \mathbb{N} \rightarrow [0, 1]$ arbitrary functions, and $p$ a polynomial function. Then $L \in \operatorname{BQP}_p(\alpha, \beta)$ if there exists a polynomial-time generated family $Q  = \{Q_n : n \in \mathbb{N}\}$ of unitary circuits, where each circuit $Q_n$
\begin{itemize}
    \item takes $n + p(n)$ input qubits -- the first $n$ qubits are used for the input $x \in L$, and the next $p(n)$ input qubits are extra ancilla qubits that the verifier is allowed,
    \item produces as output one decision qubit labeled by~$D$ and $n + p(n) - 1$ garbage qubits labeled by~$G$.
\end{itemize} 
In what follows, we write each $Q_n$ as $Q_{SA \to DG}$, thereby suppressing the dependence on the input length $n = |x|$ and explicitly indicating the systems involved at the input and output of the unitary. In addition, the circuit~$Q_n$ has the following properties:
\begin{enumerate}
    \item Completeness: For all $x \in L_{\operatorname{yes}}$,
        \begin{align}
            \label{eq:BQP-completeness}
             \Pr[Q &\text{ accepts }  x] \nonumber \\ 
             &\coloneqq \left \| (\bra{1}_D \otimes I_G) Q_{SA \to DG}(\ket{x}_S \otimes \ket{0}_A)\right\|_2^2 \nonumber \\
             &\geq \alpha (\left\vert x \right\vert).
        \end{align}
    \item Soundness: For all $x \in L_{\operatorname{no}}$,
        \begin{equation}
            \label{eq:BQP-soundness}
            \Pr[Q \text{ accepts } x ] \leq \beta (\left\vert x \right\vert),
        \end{equation}
\end{enumerate}
where acceptance is defined as obtaining the outcome one upon measuring the decision qubit register $D$ of the state $Q_{SA \to DG}(\ket{x}_S \otimes \ket{0}_A)$. Then $\operatorname{BQP} = \bigcup_p \operatorname{BQP}_p(2/3, 1/3)$, where the union is over every polynomial-bounded function $p$.

\subsection{QIP}

\label{subsec:QIP}

Quantum interactive proof systems (QIP) denote a powerful complexity class in quantum computational complexity theory. Indeed, a landmark result of the field is QIP = PSPACE \cite{jain2010qip}. The interactive proof system model involves messages between a computationally-bounded verifier and a prover with limitless computational power. These interactions may consist of some number of rounds $m$, in which case these models can be classified by the number of exchanges as QIP($m$). After all the messages have been exchanged, the verifier makes a decision to either accept or reject based on these interactions. Thus, the class QIP refers to all such promise problems that can be framed in this manner.

\begin{figure*}
    \centering
    \includegraphics[width=0.75\linewidth]{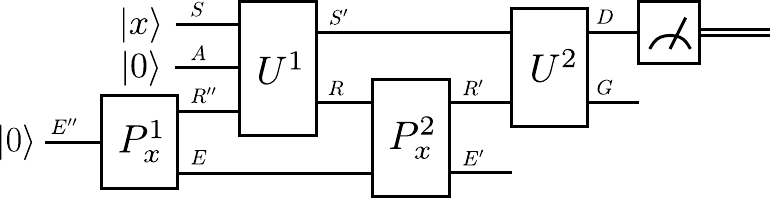}
    \caption{A general QIP(3) algorithm. The unitaries $P^1_x$ and $P^2_x$ are implemented by an all-powerful prover, and the probability of measuring the decision qubit to be in the state $\ket{1}$ is the acceptance probability of the algorithm.}
\label{fig:QIP(3)}
\end{figure*}

More formally, the definition both of QIP($m$) and QIP are given in \cite{kitaev2000parallelization,watrous2003pspaceqip} to be as follows: \
Let $m \in \mathbb{N}$, and let $\alpha,\beta :\mathbb{N} \rightarrow [0,1]$ be functions. Then let QIP$(m,\alpha,\beta)$ denote the class of promise problems $L$ for which there exists an $m$-message verifier $V$ such that
\begin{enumerate}
    \item for all $x \in L$, $\exists$ a prover $P$ such that the pair $(V,P)$ accepts with probability at least $\alpha(|x|)$, and,
    \item for all $x \notin L$, $\forall$ provers $P$, the pair $(V,P)$ accepts with probability at most $\beta(|x|)$.
\end{enumerate}
Usually, interactive proof classes are denoted solely by the number of messages exchanged, $\operatorname{QIP}(m)$. An important finding for QIP is that $\operatorname{QIP} = \operatorname{QIP}(3)$, which implies that no further computational power is afforded by increasing the number of messages exchanged beyond three \cite{kitaev2000parallelization}. A general QIP(3) algorithm can be seen in Figure~\ref{fig:QIP(3)}.

The problem of close images was the first QIP-Complete problem to be proposed \cite{kitaev2000parallelization}, and it is stated as follows:
\begin{definition}[Problem of Close Images]
\label{def:prob_close_images}
For constants $0 \leq \beta < \alpha \leq 1$, the input consists of two polynomial-time computable quantum circuits that agree on the number of output qubits and realize the quantum channels $\mathcal{N}_1$ and $\mathcal{N}_2$. Decide which of the following holds:
\begin{align}
    \textit{Yes: } & \quad \max_{\rho_1, \rho_2} F(\mathcal{N}_1(\rho_1),\mathcal{N}_2(\rho_2)) \geq \alpha, \\
    \textit{No: } & \quad  \max_{\rho_1, \rho_2} 
    F(\mathcal{N}_1(\rho_1),\mathcal{N}_2(\rho_2)) \leq \beta,
\end{align}
where the optimization is over all input states $\rho_1$ and $\rho_2$.
\end{definition}

\subsection{QMA}

The quantum Merlin--Arthur (QMA) class is equivalent to QIP(1); that is, this model consists of a single message exchanged between a computationally unbounded prover and a computationally limited verifier.

The definition of QMA can be found in \cite{watrous2009complexity}, reproduced here for convenience. Let $L = (L_{\operatorname{yes}}, L_{\operatorname{no}})$ be a promise problem, let $p, q$ be polynomially-bounded functions, and let $\alpha, \beta : \mathbb{N} \rightarrow [0, 1]$ be functions. Then $L \in \operatorname{QMA}_{p, q}(\alpha, \beta)$ if  there exists a polynomial-time generated family of unitary circuits $Q  = \{Q_n : n \in \mathbb{N}\}$, where each circuit $Q_n$ 
\begin{itemize}
    \item takes $n + p(n) + q(n)$ input qubits -- the first $n$ qubits are used for the input $x \in L$,  the next $p(n)$ input qubits are extra ancilla qubits that the verifier is allowed, and the last $q(n)$ qubits are given by the prover,
    \item produces as output one decision qubit labeled by~$D$ and $\left. n + p(n) + q(n) - 1\right.$ garbage qubits labeled by~$G$.
\end{itemize} 
As before, we write $Q_n$ as $Q_{SAP \to DG}$, thereby suppressing the dependence on the input length $n = |x|$ and explicitly indicating the systems involved at the input and output of the unitary. In addition, the circuit~$Q_n$ has the following properties:
\begin{enumerate}
    \item Completeness: For all $x \in L_{\operatorname{yes}}$, there exists a $q(\vert x\vert )$-qubit state $\sigma_P$ such that
    \begin{align}
         \Pr[Q \text{ accepts }(x, \sigma)]  & =
         \langle 1 \vert_D \operatorname{Tr}_G [\omega_{DG}] \vert 1 \rangle_D \label{eq:QMA_completeness} \\
        & \geq \alpha (\vert x\vert ),
    \end{align}
    where
    \begin{equation}
    \omega_{DG} \coloneqq Q_n (\outerproj{x}_S \otimes \outerproj{0}_A \otimes \sigma_P) Q_n^\dagger .
    \end{equation}
    \item Soundness: For all $x \in L_{\operatorname{no}}$ and every $q(\vert x\vert )$-qubit state $\sigma_P$, the following inequality holds:
    \begin{equation}
        \label{eq:QMA_soundness}
        \Pr[Q \text{ accepts }(x, \sigma)] \leq \beta (\vert x\vert ).
    \end{equation}
\end{enumerate}
Then $\operatorname{QMA} = \bigcup_{p, q} \operatorname{QMA}_{p, q}(2/3, 1/3)$, where the union is over all polynomial-bounded functions $p$ and $q$. 

\subsection{QSZK}

The complexity class quantum statistical zero-knowledge (QSZK) gives a quantum analog of the classical statistical zero-knowledge class \cite{watrous2002qszk, Wat06}, which can be phrased in terms of an interactive proof system.

We reproduce the definition of QSZK here for convenience. Let $V$ be a verifier and $P$ a prover that acts on some input $x$. Define the mixed state of the verifier and message qubits after $j$ messages to be $\rho_{V,M}(x,j)$. Then the pair $(V,P)$ is a quantum statistical zero-knowledge proof system for a promise problem $L$ if
\begin{enumerate}
    \item $(V,P)$ is an interactive proof system for $L$, and 
    \item there exists a polynomial-time preparable set $\{\sigma_{x,j}\}_j$ of states such that 
    \begin{equation}
        x \in L \Rightarrow \forall j \quad \left\Vert \sigma_{x,j} - \rho_{V,M}(x,j)  \right\Vert_{1} \leq \delta(|x|),  
    \end{equation}
    for some $\delta$ such that $\delta(n) < 1/p(n)$ for sufficiently large $n$ and every polynomial $p$.
\end{enumerate}
The completeness and soundness requirement of this class comes from the underlying proof system; for the definition of QSZK, we restrict the completeness and soundness errors to be at most $1/3$.

For this class, as with many class definitions, it can be helpful to look at a QSZK-Complete promise problem. The quantum state distinguishability problem was originally proposed alongside the class definition in \cite{watrous2002qszk}, and so it is a natural choice. The problem statement is as follows:

\begin{definition}[Quantum State Distinguishability]
\label{def:QSDProblem}
Let $L = (L_{\operatorname{yes}}, L_{\operatorname{no}})$ be a promise problem, and let $\alpha$ and $\beta$ be constants satisfying $0 \leq \beta < \alpha \leq 1$. Given two quantum circuits $Q_0$ and $Q_1$ acting on $m$ qubits each and having $k$ specified output qubits, let $\rho_i$ denote the output mixed state obtained by running $Q_i$ on an input $\ket{0}^{\otimes m}$. Decide whether
\begin{align}
\textit{Yes: } & \quad \frac{1}{2}\left\Vert \rho_0 - \rho_1 \right\Vert_1 \geq \alpha,\\
\textit{No: } & \quad \frac{1}{2}\left\Vert \rho_0 - \rho_1 \right\Vert_1 \leq \beta.
\end{align}
\end{definition}

In \cite{watrous2002qszk}, it was shown that $(\alpha,\beta)$-Quantum State Distinguishability is QSZK-Complete for $\left. 0 \leq \beta < \alpha^2 \leq 1\right. $.

\subsection{QIP(2)}

The complexity class QIP(2) specifically denotes a set of promise problems that include two messages exchanged between the prover and verifier. The formal definition can be inferred from Section~\ref{subsec:QIP} by setting the number of messages to two, i.e., $m=2$. A general QIP(2) algorithm can be seen in Figure~\ref{fig:QIP(2)}.

\begin{figure*}
    \centering
    \includegraphics[width=0.75\linewidth]{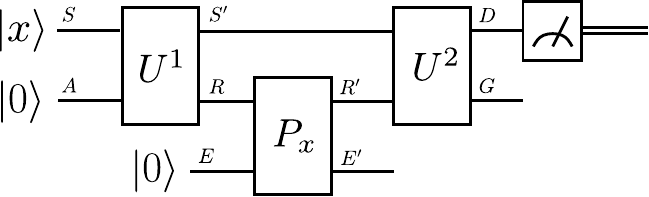}
    \caption{A general QIP(2) algorithm. The unitary $P_x$ is implemented by an all-powerful prover, and the probability of measuring the decision qubit to be in the state $\ket{1}$ is the acceptance probability of the algorithm.}
\label{fig:QIP(2)}
\end{figure*}

We now reproduce the canonical QIP(2)-Complete problem \cite{W02QIP,HMW13} as follows:
\begin{definition}[Problem of Close Image]
\label{def:prob-close-image}
Given is a circuit to realize a unitary extension $U_{AE^{\prime}\rightarrow BE}$ of a channel $\mathcal{N}_{A\rightarrow B}$, such that
\begin{multline}
\mathcal{N}_{A\rightarrow B}(\omega_{A}) =
\\
\operatorname{Tr}_{E}[U_{AE^{\prime}\rightarrow BE} (\omega_{A}\otimes\ket{0}\!\bra{0}_{E^{\prime}}) (U_{AE^{\prime}\rightarrow BE})^{\dag}]
\end{multline}
for every input state $\omega_{A}$, and a circuit to realize a purification of the state $\rho_{B}$. Decide which of the following holds:
\begin{align}
    \textit{Yes: } & \quad \max_{\sigma_{A}}F(\rho_{B},\mathcal{N}_{A\rightarrow B}(\sigma_{A})) \geq \alpha ,\\
    \textit{No: } & \quad \max_{\sigma_{A}}F(\rho_{B},\mathcal{N}_{A\rightarrow B}(\sigma_{A})) \leq \beta,
\end{align}
where the optimization is over every input state $\sigma_{A}$.
\end{definition}

Note that the Problem of Close Image is different from the Problem of Close Images (see Definition~\ref{def:prob_close_images}). In the former, we bound the fidelity between a channel and a state, whereas in the latter, we bound the fidelity between two channels.

\subsection{\texorpdfstring{\qipeb}{QIP-EB(2)}}

The complexity class \qipeb\ was introduced in  \cite{philip2023quantum} and represents a modification of QIP(2). By inspecting Figure~\ref{fig:QIP(2)} and recalling the Stinespring dilation theorem (see, e.g., \cite{wilde_2017}), we see that the prover's action in a QIP(2) protocol is equivalent to performing a quantum channel that has input system~$R$ and output system $R'$ (see also \cite[Figure~1]{JUW09}). The idea behind \qipeb\  is that the prover is constrained to performing an entanglement-breaking channel. Such a channel has the following form \cite{HSR03}:
\begin{equation}
\rho  \to \sum_{x} \Tr[\mu_x \rho]\phi_x,
\label{eq:EB-ch-def}
\end{equation}
where $\{\mu_x\}_x$ is a rank-one positive operator-valued measure (i.e., each $\mu_x$ is a rank-one positive semi-definite operator and $\sum_x \mu_x = I$) and $\{\phi_x\}_x$ is a set of pure states.

The canonical \qipeb-Complete problem is as follows \cite[Theorem~11]{philip2023quantum}:

\begin{definition}
\label{def:prob-close-image-EB-style}
Given circuits to generate a unitary extension of a channel $\mathcal{N}_{G\rightarrow S}$ and a purification of a state~$\rho_{S}$,
decide which of the following holds:
\begin{align}
    \textit{Yes: } & \ \max_{\substack{\{  (p(x),\psi^{x})\}  _{x},\\
            \left\{  \varphi^{x}\right\}_{x},\\
            \sum_{x}p(x)\psi_{S}^{x}=\rho_{S}
            }}   \sum_{x}p(x)F(\psi_{S}^{x},\mathcal{N}_{G\rightarrow S}(\varphi_{G}^{x})) \geq \alpha ,\\
    \textit{No: } & \ \max_{\substack{\{  (p(x),\psi^{x})\}  _{x},\\
            \left\{  \varphi^{x}\right\}_{x},\\
            \sum_{x}p(x)\psi_{S}^{x}=\rho_{S}
            }}   \sum_{x}p(x)F(\psi_{S}^{x},\mathcal{N}_{G\rightarrow S}(\varphi_{G}^{x})) \leq \beta,
\end{align}
        where the optimization is over every pure-state decomposition of $\rho_{S}$, as $\sum_{x}p(x)\psi_{S}^{x}=\rho_{S}$, and $\left\{  \varphi^{x}\right\}_{x}$ is a set of pure states.
\end{definition}

\subsection{QAM}

The quantum Arthur--Merlin (QAM) class was introduced in \cite{marriott2005quantum}, and it can be understood as a variation of QMA in which the verifier and prover are given access to shared randomness in advance. It can also be understood as a restricted version of QIP(2) in which the first message of the verifier is restricted to being a uniformly random classical bitstring. As such, the following containments hold: QMA $\subseteq$ QAM $\subseteq$ QIP(2).

Let us recall its definition here. Let $L = (L_{\operatorname{yes}}, L_{\operatorname{no}})$ be a promise problem, let $p, q, r$ be polynomially-bounded functions, and let $\alpha, \beta : \mathbb{N} \rightarrow [0, 1]$ be functions. Then $L \in \operatorname{QAM}_{p, q, r}(\alpha, \beta)$ if there exists a polynomial-time generated family of unitary circuits $Q  = \left\{Q_{n,y} : n \in \mathbb{N}, y \in \mathcal{Y}\right\}$, where $y$ is a uniformly random bitstring consisting of $r(n)$ bits, so that $\log_2|\mathcal{Y}| = r(n)$, and each circuit~$Q_{n,y}$ 
\begin{itemize}
    \item takes $n + p(n) + q(n)$ input qubits -- the first $n$ qubits are used for the input $x \in L$,  the next $p(n)$ input qubits are extra ancilla qubits that the verifier is allowed, and the last $q(n)$ qubits are given by the prover,
    \item produces as output one decision qubit labeled by~$D$ and $\left. n + p(n) + q(n) - 1\right.$ garbage qubits labeled by~$G$.
\end{itemize} 
We write $Q_{n,y}$ as $Q^y_{SAP \to DG}$, thereby suppressing the dependence on the input length $n = |x|$ and explicitly indicating the systems involved at the input and output of the unitary. We also use the shorthand $Q^y \equiv Q^y_{SAP \to DG}$. In addition, each set~$\{Q_{n,y}\}_{y\in\mathcal{Y}}$ of circuits has the following properties:
\begin{enumerate}
    \item Completeness: For all $x \in L_{\operatorname{yes}} $, there exists a set $\{\sigma^y_P\}_{y\in \mathcal{Y}}$ of $q(\vert x\vert )$-qubit states  such that
    \begin{align}
         & \frac{1}{|\mathcal{Y}|} \sum_{y \in \mathcal{Y}} \Pr[Q^y \text{ accepts }(x, \sigma^y)]  \notag \\
         & =
         \frac{1}{|\mathcal{Y}|} \sum_{y \in \mathcal{Y}} \langle 1 \vert_D \operatorname{Tr}_G [\omega^y_{DG}] \vert 1 \rangle_D \label{eq:QAM_completeness} \\
        & \geq \alpha (\vert x\vert ),
    \end{align}
    where
    \begin{equation}
    \omega^y_{DG} \coloneqq Q^{y} (\outerproj{x}_S \otimes \outerproj{0}_A \otimes \sigma^y_P) (Q^{y})^\dagger .
    \end{equation}
    
    \item Soundness: For all $x \in L_{\operatorname{no}}$,  and every set $\{\sigma^y_P\}_{y\in \mathcal{Y}}$ of $q(\vert x\vert )$-qubit states, the following inequality holds:
    \begin{equation}
        \label{eq:QAM_soundness}
        \frac{1}{|\mathcal{Y}|} \sum_{y \in \mathcal{Y}} \Pr[Q^y \text{ accepts }(x, \sigma^y)] \leq \beta (\vert x\vert ).
    \end{equation}
\end{enumerate}
The acceptance probability
\begin{equation}
\frac{1}{|\mathcal{Y}|} \sum_{y \in \mathcal{Y}} \Pr[Q^y \text{ accepts }(x, \sigma^y)]    
\end{equation}
can be understood as the probability of acceptance conditioned on a fixed value of $y$, which is then averaged over the shared uniform randomness (i.e., here we are applying the law of total probability). Then $\operatorname{QAM} = \bigcup_{p,q,r} \operatorname{QAM}_{p, q, r}(2/3, 1/3)$, where the union is over all polynomial-bounded functions $p$, $q$, and $r$.

\section{Results: Symmetry-Testing Problems and Quantum Computational Complexity}
\label{sec:results}

In this section, we present the main results of our work connecting the quantum complexity classes hierarchy and different symmetry testing algorithms. The results presented here all have a similar proof idea -- prove that the promise problem is in the complexity class of interest and that the problem is hard for the same complexity class. For all the hardness results, we either map an existing complete problem to the problem of interest, or we show that a generic problem that defines the class can be rewritten in terms of the problem of interest. In the latter case, we select the group and the input state such that the acceptance probability of the problem maps to the corresponding symmetry quantity. In most cases, we pick the group to be $C_2 = \{I, V\}$, with $V^2 = I$. We find that this is the simplest choice of the group $G$. The choice of $V$ is then ``reverse-engineered" in a way to match the symmetry quantity. While this preview discussion might seem a bit abstract as of now, we make the concepts more concrete in the specific examples that follow, with the first being discussed in more detail around \eqref{eq:GBS_choice_V}--\eqref{eq:GBS_Projection}.

\subsection{Testing \texorpdfstring{$G$}{G}-Bose Symmetry of a State is BQP-Complete}
\label{sec:GBS-BQP}

In this section, we show that testing the $G$-Bose Symmetry of a state is BQP-Complete. We begin now by specifying this problem statement in precise terms. 

\begin{problem}
[$\left( \alpha,\beta\right)  $-State-$G$-Bose-Symmetry]Let $\alpha$ and $\beta$ be such that $0\leq\beta<\alpha\leq1$. Given are descriptions of a circuit $U_{RS}^{\rho}$ that generates a purification of a state $\rho_S$ and circuit descriptions of a representation $\{U_S(g)\}_{g \in G}$ of a group $G$. Decide which of the following holds:
\begin{align}
\textit{Yes: }  & \quad \Tr\! \left[\Pi^G_S \rho_S \right] \geq \alpha,\\
\textit{No: }  &  \quad \Tr\! \left[\Pi^G_S \rho_S \right] \leq \beta,
\end{align}
where the group representation projector $\Pi^G_S$ is defined in~\eqref{eq:group_proj_GBS}.
\end{problem}

As observed in \cite[Section~3.1]{LRW22}, the measure $\Tr\! \left[\Pi^G_S \rho_S \right]$ is a faithful symmetry measure, in the sense that it is equal to one if and only if the state $\rho_S$ is Bose-symmetric.

\begin{theorem}
\label{thm:BQP-GBS} The promise problem State-$G$-Bose-Symmetry is BQP-Complete.

\begin{enumerate}
\item $\left(  \alpha,\beta\right)  $-State-$G$-Bose-Symmetry is in BQP for all $\beta<\alpha$, whenever the gap between $\alpha$ and $\beta$ is larger than an inverse polynomial in the input length.

\item $\left(1-  \varepsilon, \varepsilon\right)  $-State-$G$-Bose-Symmetry is BQP-Hard, even when $\varepsilon$ decays exponentially in the input length.
\end{enumerate}
Thus, $\left(  \alpha,\beta\right)  $-State-$G$-Bose-Symmetry is BQP-Complete for all $\left(  \alpha,\beta\right)  $ such that $0\leq \beta<\alpha\leq 1$.
\end{theorem}

\begin{remark}
    In the statement of Theorem~\ref{thm:BQP-GBS}, the first part indicates the largest range of parameters for which we can show that the problem is contained in BQP. Similarly, the second part indicates the largest range of parameters for which we can show that the problem is BQP-hard. Both of these parameter ranges include the case when $\alpha$ and $\beta$ are constants. As such, this leads to the final statement above that the problem is BQP-complete for constant values of $\alpha$ and $\beta$ satisfying the inequality constraint given. We present all subsequent theorems in a similar way.
\end{remark}

\begin{proof}[Proof of Theorem~\ref{thm:BQP-GBS}]
To show that the problem is BQP-Complete, we need to demonstrate two facts: first, that the problem is in BQP, and second, that it is BQP-Hard. Let us begin by proving that the problem is in BQP. In \cite[Chapter~8]{harrow2005applications} (see also \cite[Algorithm 1]{LRW22}), an algorithm was proposed to test for $G$-Bose symmetry of a state $\rho_S$ given a circuit description of unitary that generates a purification of the state and circuit descriptions of a unitary representation of a group~$G$, $\{U(g)\}_{g \in G}$. Since the algorithm can be performed efficiently, the problem is contained in~BQP. 

Next, we show that any problem in the BQP class can be reduced to an instance of this problem. The acceptance and rejection probabilities are encoded in the state of the decision qubit, with zero indicating acceptance. Now, we need to map this problem to an instance of State-$G$-Bose-Symmetry; i.e., using the circuit descriptions for a general BQP algorithm, we need to define a state~$\rho_{S'}$ and a unitary representation $\{U_{S'}(g)\}_{g \in G}$, and also show how the symmetry-testing condition $\operatorname{Tr}[\Pi_{S'}^G \rho_{S'}]$ can be written in terms of the BQP algorithm's acceptance probability. To this end, we define the group $G$ to be the cyclic group on two elements $C_2 = \{I, V\}$ such that $V^2 = I$, where $V$ is simply given by
\begin{equation}
    \label{eq:GBS_choice_V}
    V_D = -Z_D,
\end{equation}
and the input state to be 
\begin{equation}
    \rho_{D} = \Tr_{G} [Q_{SA \to DG} (\outerproj{x}_{S} \otimes \outerproj{0}_{A}) (Q_{SA \to DG})^\dagger].
\end{equation}
As such, we are making the identification $S' \leftrightarrow D$ between the system label $S'$ of a general symmetry-testing problem and the system $D$ for a BQP algorithm. The group representation and circuit to generate $\rho_D$ are thus efficiently implementable.
Furthermore, the state of interest is just the state of the decision qubit, and the group projector for the  unitary representation above is given by
\begin{equation}
\label{eq:GBS_Projection}
    \Pi^G_D = \frac{1}{2} (I_D - Z_D) = \outerprod{1}{1}_D.
\end{equation}

Furthermore, we find that the symmetry-testing condition $\Tr[\Pi_D^G \rho_D]$ maps to the BQP algorithm's acceptance probability as follows:
\begin{align}
& \Tr[\Pi_D^G \rho_D] \notag \\
& = 
\Tr[\outerproj{1}_D \rho_D] \notag\\
& = \Tr[(\outerproj{1}_D \otimes I_G) \times \notag \\
& \qquad (Q_{SA \to DG} (\outerproj{x}_{S} \otimes \outerproj{0}_{A}) (Q_{SA \to DG})^\dagger)] \notag\\
& = \left\Vert (\bra{1}_D \otimes I_G) Q_{SA \to DG} \ket{x}_S \ket{0}_A \right\Vert^2_2.
\end{align}
Comparing with \eqref{eq:BQP-completeness}, we observe that the acceptance probability of the BQP algorithm exactly matches the symmetry-testing condition of the constructed $G$-Bose symmetry-testing problem. As such, we have proven that any BQP problem can be efficiently mapped to a $G$-Bose symmetry test, concluding the proof.
\end{proof}

\subsection{Testing \texorpdfstring{$G$}{G}-Symmetry of a State Using Hilbert--Schmidt Norm is BQP-Complete}
\label{sec:Sym-HS-BQP}

In this section, we show that testing the $G$-symmetry of a state using the Hilbert--Schmidt norm is BQP-Complete, and it is thus emblematic of the class of problems efficiently solvable using quantum computers.

\begin{problem}
 [$\left( \alpha,\beta\right)  $-State-HS-Symmetry]
 Given are a circuit description of a unitary $U_{RS}^{\rho}$ that generates a purification of the state~$\rho_S$ and circuit descriptions of a unitary representation $\{U_S(g)\}_{g \in G}$ of a group $G$. Let $\alpha$ and $\beta$ be such that $0\leq\beta<\alpha\leq \gamma$, where
\begin{equation}
    \gamma \coloneqq 2\left(1 - \frac{1}{|G|}\right).
    \label{eq:gamma-bnd}
\end{equation}
  Decide which of the following holds:
\begin{align}
\textit{Yes: } & \quad \frac{1}{\vert G \vert} \sum\limits_{g \in G} \left\Vert [U(g), \rho] \right\Vert^2_2 \leq \beta,\label{eq:HS-norm-condition-states} \\
\textit{No: } & \quad \frac{1}{\vert G \vert} \sum\limits_{g \in G} \left\Vert [U(g), \rho] \right\Vert^2_2 \geq \alpha,
\end{align}
where the Hilbert--Schmidt norm of a matrix $A$ is defined as $\left\Vert A \right\Vert_2 \coloneqq \sqrt{\operatorname{Tr}[A^\dag A]}$.
\end{problem}

As observed in \cite[Section~1]{bandyopadhyay2023efficient}, the quantity in~\eqref{eq:HS-norm-condition-states} is a faithful symmetry measure in the sense that it is equal to zero if and only if the state $\rho$ is $G$-symmetric, as in Definition~\ref{def:G-symmetry}.

The quantity $\gamma$ in \eqref{eq:gamma-bnd} arises as a natural upper bound for $\frac{1}{\vert G \vert} \sum\limits_{g \in G} \left\Vert [U(g), \rho] \right\Vert^2_2$ because
\begin{align}
& \frac{1}{\vert G \vert} \sum\limits_{g \in G} \left\Vert [U(g), \rho] \right\Vert^2_2 \notag \\
&  = \frac{1}{\vert G \vert} \sum\limits_{g \in G} 2 \left( \Tr[\rho^2] - \Tr[\rho U(g) \rho U(g)^\dagger] \right) \notag \\
&  = \frac{1}{\vert G \vert} \sum\limits_{g \in G, g \neq e} 2 \left( \Tr[\rho^2] - \Tr[\rho U(g) \rho U(g)^\dagger] \right) \notag \\
& \leq \gamma,
\end{align}
where the first equality follows from \eqref{eq:HS-proof-step-1} below and the inequality follows because $\Tr[\rho^2] \leq 1$ and $\Tr[\rho U(g) \rho U(g)^\dagger] \geq 0$. Furthermore, the upper bound is saturated by choosing $\rho$ to be $\outerproj{0}$ in a $d$-dimensional space and the unitary representation to be the Heisenberg--Weyl shift operators $\{X(x)\}_{x=0}^{d-1}$, such that $X(x)\ket{0} = \ket{x}$.

\begin{theorem}
\label{thm:BQP-Sym-HS} The promise problem State-HS-Symmetry is BQP-Complete.

\begin{enumerate}
\item $\left(  \alpha,\beta\right)  $-State-HS-Symmetry is in BQP for all $\beta<\alpha$. (It is implicit that the gap between $\alpha$ and $\beta$ is larger than an inverse polynomial in the input length.)

\item $\left(  \gamma-\varepsilon, \varepsilon\right)  $-State-HS-Symmetry is BQP-Hard, even when $\varepsilon$ decays exponentially in the input length.
\end{enumerate}

\noindent Thus, $\left(  \alpha,\beta\right)  $-State-HS-Symmetry is BQP-Complete for all $\left(  \alpha,\beta\right)  $ such that $0\leq \beta<\alpha\leq \gamma$.
\end{theorem}

\begin{proof}
To show that the problem is BQP-Complete, we first show that it is in the BQP class and then show that it is BQP-Hard. For the first part, let us briefly recall the development from \cite[Section~3.1]{bandyopadhyay2023efficient}. Consider the following equalities:
\begin{align}
    &\left\Vert [U(g), \rho] \right\Vert^2_2 \nonumber \\
    &= \left\Vert \rho U(g) - U(g)\rho \right\Vert^2_2 \nonumber \\
    &= \left\Vert \rho - U(g) \rho U(g)^\dagger \right\Vert^2_2 \nonumber \\
    &= \Tr[\rho^2] + \Tr[(U(g) \rho U(g)^\dagger)^2] - 2\Tr[\rho U(g) \rho U(g)^\dagger]\nonumber \\
    &= 2\left( \Tr[\rho^2] - \Tr[\rho U(g) \rho U(g)^\dagger] \right),
    \label{eq:HS-proof-step-1}
\end{align}
where the second equality is due to the unitary invariance of the Hilbert--Schmidt norm. Thus, we see that
\begin{align}
    &\frac{1}{\vert G \vert} \sum\limits_{g \in G} \left\Vert [U(g), \rho] \right\Vert^2_2 \nonumber \\
    &= \frac{1}{\vert G \vert} \sum\limits_{g \in G} 2 \left( \Tr[\rho^2] - \Tr[\rho U(g) \rho U(g)^\dagger] \right) \nonumber \\
    &= 2 \left( \Tr[\rho^2] - \Tr[\rho \mathcal{T}_G(\rho)] \right) \nonumber \\
    &= 2 \left( \Tr[\operatorname{SWAP} (\rho \otimes \rho)] - \Tr[\operatorname{SWAP} (\rho \otimes \mathcal{T}_G(\rho))] \right),
\end{align}
where $\mathcal{T}_G$ is the twirl channel given by 
\begin{equation}
    \mathcal{T}_G( \cdot ) \coloneqq  \frac{1}{\vert G \vert } \sum\limits_{g \in G} U(g) (\cdot) U(g)^\dagger
\end{equation}
and SWAP is the unitary swap operator.
The two terms can be individually estimated by means of the destructive SWAP test \cite{GC13}. To realize the twirl, one can pick an element $g$ uniformly at random and apply $U(g)$ to the state $\rho$. Since the twirl and the SWAP test can be efficiently performed, it follows that the problem is in the BQP class.

Next, we show that the problem is BQP-Hard by providing an efficient mapping from a general BQP problem to our problem of interest. Consider a general BQP algorithm as described in Section~\ref{subsec:BQP}. The output state of the BQP algorithm is given by
\begin{equation}
    Q_{SA \to DG} \ket{x}_S \ket{0}_A.
\end{equation}
Then, the acceptance and rejection probabilities of the BQP algorithm are given by
\begin{align}
    p_{\operatorname{acc}} &= \left\Vert (\bra{1}_D \otimes I_G) Q_{SA \to DG} \ket{x}_S \ket{0}_A \right\Vert^2_2, \\
    p_{\operatorname{rej}} &= 1-p_{\operatorname{acc}} \nonumber \\
    &= \left\Vert (\bra{0}_D \otimes I_G) Q_{SA \to DG} \ket{x}_S \ket{0}_A \right\Vert^2_2.
\end{align}

Now, we need to map this problem to an instance of State-HS-Symmetry; i.e., we need to define a state $\rho$ and a unitary representation $\{U(g)\}_{g \in G}$. To this end, let us define the group $G$ to be the cyclic group on two elements, $C_2 = \{I, V\}$, such that $V^2 = I$, where $V$ is given by
\begin{equation}
    V_{SAC} = (Q_{SA \to DG})^\dagger \operatorname{CNOT}_{DC} Q_{SA \to DG},
\end{equation}
and the input state to be 
\begin{equation}
    \rho_{SAC} = \outerprod{x}{x}_S \otimes \outerprod{0}{0}_{AC}.
\end{equation}
From \eqref{eq:HS-proof-step-1}, we see that
\begin{align}
    &\frac{1}{\vert G \vert} \sum\limits_{g \in G} \left\Vert [U(g), \rho] \right\Vert^2_2 \nonumber \\
    &= \frac{1}{|G|}\sum\limits_{g \in G} 2(\Tr[\rho^2] - \Tr[\rho U(g) \rho U(g)^\dagger] )\nonumber \\
    &= 1 - \Tr[\rho V \rho V^\dagger] \nonumber \\
    &= 1- \left\vert (\bra{x}_S \otimes \bra{0}_{AC}) V (\ket{x}_S \otimes \ket{0}_{AC}) \right\vert^2.
\end{align}
To show the equivalence, we now expand $V$ as follows:
\begin{multline}
    V (\ket{x}_S \otimes \ket{0}_{AC}) = \\
     (Q_{SA \to DG})^\dagger \operatorname{CNOT}_{DC} Q_{SA \to DG} (\ket{x}_S \otimes \ket{0}_{AC}).
\end{multline}
Next, we insert an identity operator $I_D$ to simplify:
\begin{multline}
    (Q_{SA \to DG})^\dagger \operatorname{CNOT}_{DC} Q_{SA \to DG} (\ket{x}_S \otimes \ket{0}_{AC})  \\ 
    = (Q_{SA \to DG})^\dagger \operatorname{CNOT}_{DC} (\outerproj{0}_D \otimes I_{GC}  \\ 
     + \outerproj{1}_D \otimes I_{GC}) Q_{SA \to DG} (\ket{x}_S \otimes \ket{0}_{AC}).
\end{multline}
Expanding, this reduces to
\begin{align}
   &(Q_{SA \to DG})^\dagger (\outerproj{0}_D \otimes I_{GC} \nonumber \\ 
   &\qquad + \outerproj{1}_D \otimes I_G \otimes X_C) Q_{SA \to DG} (\ket{x}_S \otimes \ket{0}_{AC}) \nonumber \\ 
   &= (Q_{SA \to DG})^\dagger (\outerproj{0}_D) Q_{SA \to DG} (\ket{x}_S \otimes \ket{0}_{AC}) \nonumber \\ 
   &\qquad + (Q_{SA \to DG})^\dagger (\outerproj{1}_D) Q_{SA \to DG} (\ket{x}_S \otimes \ket{0}_A \ket{1}_C).
\end{align}
Thus, by expanding $p_{\operatorname{rej}}$ as
\begin{align}
    p_{\operatorname{rej}} & = \left \| (\bra{0}_D \otimes I_G) Q_{SA \to DG}(\ket{x}_S \otimes \ket{0}_A)\right\|_2^2 \notag \\
    & = (\bra{x}_S \otimes \bra{0}_A) (Q_{SA \to DG})^\dagger (\outerproj{0}_D) \times \notag \\
    & \qquad \qquad Q_{SA \to DG} (\ket{x}_S \otimes \ket{0}_{A}),
\end{align}
we find that
\begin{align}
    &(\bra{x}_S \otimes \bra{0}_{AC}) V (\ket{x}_S \otimes \ket{0}_{AC}) \nonumber \\
    &= (\bra{x}_S \otimes \bra{0}_A) (Q_{SA \to DG})^\dagger (\outerproj{0}_D) \times \notag \\
    & \qquad \qquad Q_{SA \to DG} (\ket{x}_S \otimes \ket{0}_{A}) \notag \\
    &= p_{\operatorname{rej}}.
\end{align}
We then finally see that
\begin{equation}
    \label{eq:mixed-HS-sym-accep-prob}
    q \coloneqq \frac{1}{\vert G \vert} \sum\limits_{g \in G} \left\Vert [U(g), \rho] \right\Vert^2_2 = 1 - p_{\operatorname{rej}}^2.
\end{equation}
Thus, given a method to estimate $q$ within additive error $\varepsilon$, we can estimate $p_{\operatorname{rej}}$ within an additive error of $\sqrt{\varepsilon}$. A proof of this can be found in Appendix~\ref{app:error-mixed-state-HS-symmetry}. We can then estimate $p_{\operatorname{acc}} = 1 - \sqrt{1-q}$ within an additive error of $\sqrt{\varepsilon}$ as well.  As such, a general BQP problem can be efficiently mapped to our problem of interest, showing that State-HS-Symmetry is BQP-Hard. This along with that the fact that the problem lies in BQP, completes the proof of BQP-Completeness.
\end{proof}

\subsection{Testing \texorpdfstring{$G$}{G}-Bose Symmetry of the Output of a Channel is QMA-Complete}
\label{sec:GBS-QMA}

In this section, we show that testing the $G$-Bose symmetry of the output of a channel with optimized input is QMA-Complete. 

\begin{problem}
[$\left(\alpha,\beta\right)$-Channel-$G$-Bose-Symmetry]Let $\alpha$ and $\beta$ be such that $0\leq \beta < \alpha \leq1$. Given is a circuit description of a unitary $U_{BD' \to SD}^{\mathcal{N}}$ that realizes a unitary dilation of a channel
\begin{multline}
\label{eq:unitary_dilation_channel}
\mathcal{N}_{B \to S} (\cdot) \coloneqq \\
\operatorname{Tr}_D[U^\mathcal{N}_{BD' \to SD}((\cdot)_B \otimes \outerproj{0}_{D'}) (U^\mathcal{N}_{BD' \to SD})^\dagger]  
\end{multline}
and circuit descriptions of a unitary representation $\{U_S(g)\}_{g \in G}$ of a group~$G$. Decide which of
the following holds:
\begin{align}
\textit{Yes: } & \quad \max\limits_{\rho_B}\  \Tr\!\left[ \Pi^G_S \mathcal{N}_{B \to S}(\rho_{B}) \right] \geq \alpha,
\label{eq:fid-ch-pure-state}
\\
\textit{No: } & \quad \max\limits_{\rho_B}\  \Tr\!\left[ \Pi^G_S \mathcal{N}_{B \to S}(\rho_{B}) \right] \leq \beta,
\end{align}
where the optimization is over every input state $\rho_B$.
\end{problem}

Let us observe that the measure in \eqref{eq:fid-ch-pure-state} is a faithful symmetry measure, in the sense that it is equal to one if and only if there exists an input state $\rho_{B}$ such that the output state $\mathcal{N}_{B \to S}(\rho_{B})$ is Bose-symmetric. This follows from continuity of $\Tr\!\left[ \Pi^G_S \mathcal{N}_{B \to S}(\rho_{B}) \right]$ and from the arguments in \cite[Section~3.1]{LRW22}.

\begin{theorem}
\label{thm:QMA-Chl-GBS} The promise problem Channel-$G$-Bose-Symmetry is QMA-Complete.

\begin{enumerate}
\item $\left(  \alpha,\beta\right)  $-Channel-$G$-Bose-Symmetry is in QMA for all $\beta<\alpha$. (It is implicit that the gap between $\alpha$ and $\beta$ is larger than an inverse polynomial in the input length.)

\item $\left( 1-\varepsilon, \varepsilon\right)  $-Channel-$G$-Bose-Symmetry is QMA-Hard, even when $\varepsilon$ decays exponentially in the input length.
\end{enumerate}

\noindent Thus, $\left(  \alpha,\beta\right)  $-Channel-Bose-Symmetry is QMA-Complete for all $\left(  \alpha,\beta\right)  $ such that $0\leq \beta<\alpha\leq 1$.
\end{theorem}

\begin{proof}
To show that the problem is QMA-Complete, we need to demonstrate two facts: first, that the problem is in QMA, and second, that it is QMA-Hard. Let us begin by proving that the problem is in QMA. Let~$\rho_B$ be a state sent by the prover. This state is input to the channel $\mathcal{N}_{B \to S}$ defined in \eqref{eq:unitary_dilation_channel}. This leads to the output state $\mathcal{N}_{B \to S}(\rho_{B})$, on which the $G$-Bose symmetry test is conducted. From \cite[Algorithm 1]{LRW22}, we see that testing $G$-Bose symmetry can be done efficiently when circuit descriptions of the unitaries $\{U(g)\}_{g \in G}$ are provided. The acceptance probability of this test, for an input state~$\rho_B$, is equal to $\Tr\!\left[ \Pi^G_S \mathcal{N}_{B \to S}(\rho_{B}) \right]$. Since the prover can optimize this probability over all possible input states, the acceptance probability is then 
\begin{equation}
    \max\limits_{\rho_B}\  \Tr\!\left[ \Pi^G_S \mathcal{N}_{B \to S}(\rho_{B}) \right].
\end{equation} 
Thus, the entire estimation can be done efficiently when aided by an all-powerful prover, establishing that the problem lies in QMA.

To show that the problem is QMA-Hard, we pick an arbitrary QMA problem and map it to Channel-$G$-Bose-Symmetry. We use a similar construction proposed above in Section~\ref{sec:GBS-BQP}. For the mapping, we need to define a channel $\mathcal{N}$ and a unitary representation $\{U(g)\}_{g \in G}$; i.e., using the circuit descriptions for a general QMA algorithm, we need to define a channel~$\mathcal{N}_{P \to S'}$ and a unitary representation $\{U_{S'}(g)\}_{g \in G}$, and also show how the symmetry-testing condition 
\begin{equation}
    \max\limits_{\rho_{P}}\operatorname{Tr}[\Pi_{S'}^G \mathcal{N}_{P \to S'}(\rho_P)]
\end{equation} 
can be written in terms of the QMA algorithm's acceptance probability. To this end, we define the group $G$ to be the cyclic group on two elements $C_2 = \{I, V\}$ such that $V^2 = I$, where $V$ is simply given by
\begin{equation}
    V_D = -Z_D,
\end{equation}
and we define the channel $\mathcal{N}_{P \to D}$ to be 
\begin{multline}
    \mathcal{N}_{P \to D}(\cdot) \coloneqq \\
    \Tr_{G} [Q_{SAP \to DG} (\outerproj{x}_{S} \otimes\outerproj{0}_{A} \otimes (\cdot)_P) (Q_{SAP \to DG})^\dagger].
\end{multline} 
As such, we are making the identification $S' \leftrightarrow D$ between the system label $S'$ of a general symmetry-testing problem and the system $D$ for a QMA algorithm. The group representation and the channel are thus efficiently implementable.
Thus, the channel output state of interest is just the state of the decision qubit, and the group projector for the given unitary representation is given by
\begin{equation}
    \Pi^G_D = \frac{1}{2} (I_D - Z_D) = \outerprod{1}{1}_D.
\end{equation}
We then find, for a fixed input state $\rho_P$, the following equalities relating the symmetry-testing condition to the QMA algorithm's acceptance probability:
\begin{align}
    & \Tr\!\left[\Pi^G_D \mathcal{N}_{P \to D}(\rho_P)\right] \notag \\
   &  = \Tr\!\left[\outerproj{1}_D \mathcal{N}_{P \to D}(\rho_P)\right] \notag \\
   &  = \Tr [(\outerproj{1}_D \otimes I_G) (Q_{SAP \to DG} (\outerproj{x}_{A} \otimes\notag \\
   & \qquad\qquad  \outerproj{0}_{A} \otimes \rho_P) (Q_{SAP \to DG})^\dagger)]\notag  \\
   & = \Pr[Q \text{ accepts } (x,\rho_P)].
\end{align}
The prover then optimizes this probability over every input state $\rho_P$, and we observe that the acceptance probability of the QMA algorithm exactly matches the symmetry-testing condition of the constructed $G$-Bose symmetry testing algorithm. As such, we have proven that any QMA problem can be efficiently mapped to an instance of a Channel-$G$-Bose symmetry-testing problem, concluding the proof. 
\end{proof}

\subsection{Testing \texorpdfstring{$G$}{G}-Symmetry of a State using Trace Norm is QSZK-Complete}
\label{sec:GS-TD-QSZK}

In Section~\ref{sec:Sym-HS-BQP}, we showed that testing $G$-symmetry of a state using the Hilbert--Schmidt norm is BQP-Complete. In this section, we show that testing the $G$-symmetry of a state using the trace norm is QSZK-Complete. As such, the complexity of a $G$-symmetry test depends on the measure being used, much like what was observed in \cite[Section~V]{RASW23}.

\begin{problem}
[$\left( \alpha, \beta \right)  $-State-$G$-Sym-TD]
\label{prob:MinTD-QSZK}
Let $\alpha$ and $\beta$ be such that $0\leq\beta<\alpha\leq1$. Given are a circuit description of a unitary $U_{RS}^{\rho}$ that generates a purification of a state~$\rho_S$ and circuit descriptions of a unitary representation $\{U_S(g)\}_{g\in G}$ of a group~$G$. Decide which of the following holds:
\begin{align}
\textit{Yes: } & \quad \min\limits_{\sigma \in \operatorname{Sym}_G} \frac{1}{2} \left \Vert \rho - \sigma\right  \Vert_1 \geq \alpha,
\label{eq:TD-asym-meas}\\
\textit{No: }  &  \quad \min\limits_{\sigma \in \operatorname{Sym}_G} \frac{1}{2} \left\Vert \rho - \sigma \right\Vert_1 \leq \beta,
\end{align}
where the set $\operatorname{Sym}_G$ is defined as follows:
\begin{equation}
    \operatorname{Sym}_G \coloneqq \{\sigma \in \mathcal{D}(\mathcal{H}) : U(g)\sigma U(g)^\dagger = \sigma \quad \forall g \in G \}.
    \label{eq:def-G-Symm-set}
\end{equation}
\end{problem}

Let us observe that the asymmetry measure in \eqref{eq:TD-asym-meas} is faithful, in the sense that it is equal to zero if and only if the state $\rho$ is $G$-symmetric. This follows from the faithfulness of the trace norm and its continuity properties.

\begin{theorem}
\label{thm:QSZK-Min-TD} The promise problem State-$G$-Sym-TD is QSZK-Complete.
\begin{enumerate}
\item $\left( \alpha,\beta\right) $-State-$G$-Sym-TD is in QSZK for all $2\beta<\alpha$. (It is implicit that the gap between $\alpha$ and $2\beta$ is larger than an inverse polynomial in the input length.)

\item $\left( 1-\varepsilon, \varepsilon\right) $-State-$G$-Sym-TD is QSZK-Hard, even when $\varepsilon$ decays exponentially in the input length.
\end{enumerate}

\noindent Thus, $\left( \alpha,\beta\right) $-State-$G$-Sym-TD is QSZK-Complete for all $\left(  \alpha,\beta\right)  $ such that $0<2\beta<\alpha<1$.
\end{theorem}

\begin{proof}
To prove that the problem is QSZK-Complete, we need to show two facts. First, we need to show that the problem is in the QSZK class. Next, we show that the problem is QSZK-Hard; i.e., every problem in QSZK can be efficiently mapped to this problem.

First, we show that it is in QSZK. We begin with a simple calculation. Consider the case of a No instance, for which the following inequality holds
\begin{equation}
    \minSymG{\rho}{\sigma} \leq \beta.
\end{equation}
Let $\sigma^{\ast}\in$ Sym$_{G}$ be a state achieving the minimum, so that
\begin{equation}
    \label{eq:Minsig*}
    \td{\rho}{\sigma^\ast} \leq \beta.
\end{equation}
Define $\overline{\rho}$  as the result of twirling $\rho$ with respect to the group elements of $G$. More concretely,
\begin{equation}
    \overline{\rho} \coloneqq \mathcal{T}_G(\rho) = \frac{1}{\vert G \vert} \sum\limits_{g \in G} U(g) \rho U^\dagger(g).
\end{equation}
From the data-processing inequality for the trace distance \cite[Chapter 9]{wilde_2017}, and the fact that $\mathcal{T}_G(\sigma^\ast) = \sigma^\ast$, we conclude that
\begin{align}
    \td{\overline{\rho}}{\sigma^\ast} &= \td{\mathcal{T}_G(\rho)}{\mathcal{T}_G(\sigma^\ast)} \nonumber \\
    &\leq \td{\rho}{\sigma^\ast} \nonumber \\
    &\leq \beta.
    \label{eq:DPIMinSymG}
\end{align}
Now, using the triangle inequality for the trace distance, \eqref{eq:Minsig*}, and \eqref{eq:DPIMinSymG}, we find that
\begin{align}
    \frac{1}{2}\left\Vert \rho-\overline{\rho}\right\Vert _{1} & \leq\frac{1}%
    {2}\left\Vert \rho-\sigma^{\ast}\right\Vert _{1}+\frac{1}{2}\left\Vert
    \sigma^{\ast}-\overline{\rho}\right\Vert _{1} \label{eq:triangle-TD-pf-qszk-1}\\
    & \leq\beta + \beta = 2\beta.
    \label{eq:triangle-TD-pf-qszk-2}
\end{align}

Having established the above, we now construct a QSZK algorithm consisting of the following steps:
\begin{enumerate}
\item The verifier randomly prepares the state
$\rho$ or $\overline{\rho}$. The verifier can prepare the latter state by
preparing $\rho$ and performing the group twirl~$\mathcal{T}_{G}$.

\item The
verifier sends the state to the prover, who performs an optimal measurement to
distinguish $\rho$ from $\overline{\rho}$. 

\item The verifier accepts if the prover
can guess the state that was prepared, and the maximum acceptance probability of the prover is given
by \cite{helstrom_quantum_1969,holevo_analogue_1972}
\begin{equation}
\frac{1}{2}\left( 1+ \td{\rho}{\overline{\rho}}\right).
\end{equation}
\end{enumerate}

In the case of a No instance, by applying \eqref{eq:triangle-TD-pf-qszk-1}--\eqref{eq:triangle-TD-pf-qszk-2}, this probability is bounded from above as
\begin{equation}
\frac{1}{2}\left(  1+\frac{1}{2}\left\Vert \rho-\overline{\rho}\right\Vert
_{1}\right)  \leq\frac{1}{2}\left(  1+2\beta\right)  =\frac{1}%
{2}+\beta.
\end{equation}
In the case of a Yes instance, we find that%
\begin{align}
\begin{aligned}
\frac{1}{2}\left\Vert \rho-\overline{\rho}\right\Vert _{1}  & \geq\frac{1}%
{2}\left\Vert \rho-\sigma^{\ast}\right\Vert _{1}\\
& \geq \alpha.
\end{aligned}
\end{align}
This then implies, in this case, that the acceptance probability satisfies
\begin{align}
\frac{1}{2}\left(  1+\frac{1}{2}\left\Vert \rho-\overline{\rho}\right\Vert
_{1}\right)    & \geq\frac{1}{2}\left(  1+\alpha\right)  \\
& =\frac{1}{2}+\frac{1}{2}\alpha.
\end{align}
Thus, there is a gap as long as%
\begin{equation}
\frac{1}{2}+\frac{1}{2}\alpha > \frac{1}{2}+\beta,
\end{equation}
which is the same as $\alpha > 2\beta$.

The interactive proof system is quantum statistical zero-knowledge because, in the case of a Yes instance, the verifier can efficiently simulate the whole interaction on their own, and the statistical difference between the simulation and the actual protocol is negligible. Thus, the problem is in the QSZK class.

Next, we show that an arbitrary problem in the QSZK class can be efficiently mapped to this problem. We do so by mapping a known QSZK-Complete problem to this problem. We pick the $(\alpha_s, \beta_s)$-State Distinguishability Problem (see Definition~\ref{def:QSDProblem}).

\begin{figure}
    \centering
    \includegraphics[width=0.75\columnwidth]{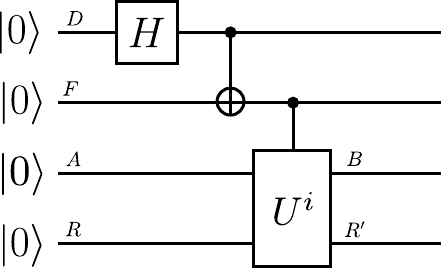}
    \caption{Circuit to create the classical--quantum state $\tau_{FB}$. The unitaries $U^0$ and $U^1$ generate purification of states $\omega^0$ and $\omega^1$, respectively. More concretely, $\left.\omega^i_B = \Tr_{R'} [U^i \outerproj{00}_{AR} (U^i)^\dagger]\right.$ for $i\in\{0,1\}$.}
    \label{fig:GenerateRhoBar}
\end{figure}

Given circuits to generate the states $\rho^0_B$ and $\rho^1_B$ with the following soundness and completeness parameters
\begin{align}
    \textit{Yes: } \quad &\td{\rho^0_B}{\rho^1_B} \geq \alpha_s, \\
    \textit{No: } \quad &\td{\rho^0_B}{\rho^1_B} \leq \beta_s,
\end{align}
we use the construction from \cite[Theorem~1]{watrous2002qszk} to create circuits that generate states $\omega^0$, $\omega^1$ such that
\begin{align}
    \textit{Yes: } \quad &\td{\omega^0}{\omega^1} \geq 1-2^{-n} \\
    \textit{No: } \quad &\td{\omega^0}{\omega^1} \leq 2^{-n}.
    \label{eq:NO-instance-wat-map}
\end{align}

The value of $n$ will be chosen later in the proof. The procedure from \cite[Theorem~1]{watrous2002qszk} runs in time polynomial in $n$ and the size of the circuits that generate $\rho^i$, and thus it is efficient.

Next, we define the following state
\begin{equation}
    \tau_{FB} \coloneqq \frac{1}{2} \left(\outerproj{0}_F \otimes \omega^0_B + \outerproj{1}_F \otimes \omega^1_B \right),
\end{equation}
and define the group $G$ to be $\{I_F \otimes I_B, X_F \otimes I_B\}$. A circuit to create the state~$\tau_{FB}$ is given in Figure~\ref{fig:GenerateRhoBar}. Twirling this state with respect to the group elements results in the state 
\begin{equation}
    \overline{\tau}_{FB} \coloneqq \mathcal{T}_G(\tau_{FB}) = \pi_F \otimes \frac{1}{2} (\omega^0_B + \omega^1_B),
\end{equation}
where $\pi_F\coloneqq \frac{1}{2}(\outerproj{0}_F + \outerproj{1}_F).$
Then we find that
\begin{align}
& \tau_{FB}- \overline{\tau}_{FB} \notag \\
& = \frac{1}{2} \outerproj{0} \otimes \omega^0 + \frac{1}{2} \outerproj{1} \otimes \omega^1 - \pi \otimes \frac{1}{2} \left( \omega^0 + \omega^1 \right) \notag\\
& = \frac{1}{2} \outerproj{0} \otimes \omega^0 + \frac{1}{2} \outerproj{1} \otimes\omega^1 \notag \\
& \qquad - \left( \frac{1}{2} \outerproj{0} + \frac{1}{2} \outerproj{1} \right)  \otimes\frac{1}{2} \left(\omega^0 + \omega^1 \right) \notag\\
& = \frac{1}{2} \outerproj{0} \otimes\left(  \omega^0 - \frac{1}{2}\left( \omega^0 + \omega^1 \right) \right) \notag \\
& \qquad + \frac{1}{2} \outerproj{1} \otimes \left( \omega^1 - \frac{1}{2} \left( \omega^0 + \omega^1 \right)  \right) \notag\\
& = \frac{1}{2} \outerproj{0} \otimes\frac{1}{2} \left( \omega^0 - \omega^1 \right)  + \frac{1}{2} \outerproj{1} \otimes \frac{1}{2} \left(  \omega^1 - \omega^0 \right),
\end{align}
which implies that
\begin{equation}
\label{eq:QSD-MinTD-mapping}
    \td{\tau_{FB}}{\overline{\tau}_{FB}} 
     =\frac{1}{4}\left\Vert \omega^0 - \omega^1 \right\Vert_{1}.
\end{equation}

We now map Yes instances of Quantum-State-Distinguishability to Yes instances of State-$G$-Sym-TD. 
For a Yes instance,
\begin{equation}
    \td{\rho^0}{\rho^1} \geq \alpha_s \implies \td{\omega^0}{\omega^1} \geq 1-2^{-n}.
\end{equation}
Define $\sigma^\ast$ to be a state that achieves the following minimum:
\begin{equation}
    \td{\tau_{FB}}{\sigma^\ast} \coloneqq \minSymG{\tau_{FB}}{\sigma}.
\end{equation}
Using the triangle inequality and the data-processing inequality (the latter being similar to how it was used before in \eqref{eq:DPIMinSymG}), we find that
\begin{align}
    \td{\tau_{FB}}{\overline{\tau}_{FB}} &\leq \td{\tau_{FB}}{\sigma^\ast} + \td{\sigma^\ast}{\overline{\tau}_{FB}} \nonumber \\
    &\leq 2\left( \td{\tau_{FB}}{\sigma^\ast} \right) \nonumber \\
    &= 2\left( \minSymG{\tau_{FB}}{\sigma} \right).
\end{align}
Thus, using \eqref{eq:QSD-MinTD-mapping}, we see that
\begin{align}
    \minSymG{\tau_{FB}}{\sigma} &\geq \frac{1}{2} \left(\td{\tau_{FB}}{\overline{\tau}_{FB}} \right) \nonumber \\
    &= \frac{1}{4} \left( \td{\omega^0}{\omega^1} \right) \nonumber \\
    &\geq \frac{1-2^{-n}}{4}.
\end{align}
As such, the Yes instances are mapped as follows:
\begin{equation}
    \td{\rho^0}{\rho^1} \geq \alpha_s \Rightarrow \minSymG{\tau_{FB}}{\sigma} \geq \frac{1-2^{-n}}{4}.
\end{equation}

Similarly, consider a NO instance of Quantum-State-Distinguishability,
\begin{equation}
    \td{\rho^0}{\rho^1} \leq \beta_s.
\end{equation}
Then using \eqref{eq:NO-instance-wat-map} and \eqref{eq:QSD-MinTD-mapping},
\begin{align}
    \minSymG{\tau_{FB}}{\sigma} &\leq \td{\tau_{FB}}{\overline{\tau}_{FB}} \nonumber \\
    &= \frac{1}{2} \left(\td{\omega^0}{\omega^1} \right) \nonumber \\
    &\leq 2^{-n-1} .
\end{align}
As such, we have shown that $\left(\alpha_s, \beta_s \right)  $-Quantum-State-Distinguishability is efficiently mapped to $\left( \frac{1}{4}(1-2^{-n}), 2^{-n-1} \right)  $-State-$G$-Sym-TD. Thus, a gap exists between the soundness and completeness conditions if
\begin{align}
    \frac{(1-2^{-n})}{4} &> 2^{-n-1} ,
\end{align}
which is equivalent to $n > \operatorname{log}_2(3)$.

To conclude that State-$G$-Sym-TD is QSZK-Complete for arbitrary constants $\alpha$ and $\beta$, one can use the constructions from Lemmas~2 and 3 of \cite{watrous2002qszk} to manipulate the parameters $\frac{1}{4}(1-2^{-n})$ and $ 2^{-n-1}$ as desired. The reasoning for this latter statement is similar to that given at the end of the proof of \cite[Theorem~6]{watrous2002qszk}.
\end{proof}

\subsection{Testing \texorpdfstring{$G$}{G}-Symmetry of a State using Fidelity is QSZK-Complete}
\label{sec:GS-Fid-QSZK}

In this section, we show that testing $G$-Symmetry of a state using fidelity is QSZK-Complete, where the fidelity of states $\rho$ and $\sigma$ is defined as  \cite{Uhl76}
\begin{equation}
F(\rho,\sigma)\coloneqq \left \Vert \sqrt{\rho} \sqrt{\sigma} \right \Vert_1^2 .    
\end{equation}
To show hardness, we provide an efficient mapping from the problem State-$G$-Sym-TD, defined in Section~\ref{sec:GS-TD-QSZK}, to the problem of interest State-$G$-Sym-Fid. 

\begin{problem}
[$\left( \alpha,\beta\right)  $-State-$G$-Sym-Fid]
\label{prob:GSym-QSZK}
Let $\alpha$ and $\beta$ be such that $0\leq\beta<\alpha\leq1$. Given are a circuit description of a unitary $U_{RS}^{\rho}$ that generates a purification of a state~$\rho_S$ and circuit descriptions of a unitary  representation $\{U(g)\}_{g \in G}$ of a group $G$. Decide which of the following holds:
\begin{align}
\textit{Yes: } & \quad \max\limits_{\sigma \in \operatorname{Sym}_G} F(\rho, \sigma) \geq \alpha,
\label{eq:fid-sym-meas-yes} \\
\textit{No: } & \quad \max\limits_{\sigma \in \operatorname{Sym}_G} F(\rho, \sigma) \leq \beta,
\end{align}
where the set $\operatorname{Sym}_G$ is defined in \eqref{eq:def-G-Symm-set}.
\end{problem} 

Let us observe that the symmetry measure in \eqref{eq:fid-sym-meas-yes} is faithful, in the sense that it is equal to one if and only if the state $\rho$ is $G$-symmetric. This follows from the faithfulness and continuity properties of fidelity.

\begin{theorem}
\label{thm:QSZK-GSym} The promise problem State-$G$-Sym-Fid is QSZK-Complete.
\begin{enumerate}
\item $\left( \alpha,\beta\right) $-State-$G$-Sym-Fid is in QSZK for all $\beta<4\alpha-3$. (It is implicit that the gap between $4\alpha - 3$ and $\beta$ is larger than an inverse polynomial in the input length.)

\item $\left( 1-\varepsilon, \varepsilon\right) $-State-$G$-Sym-Fid is QSZK-Hard, even when $\varepsilon$ decays exponentially in the input length.
\end{enumerate}
Thus, $\left( \alpha,\beta\right) $-State-$G$-Sym-Fid is QSZK-Complete for all $\left(  \alpha,\beta\right)  $ such that $0\leq \beta<4\alpha-3\leq 1$.
\end{theorem}

\begin{figure*}
    \centering    
    \includegraphics[width=0.9\textwidth]{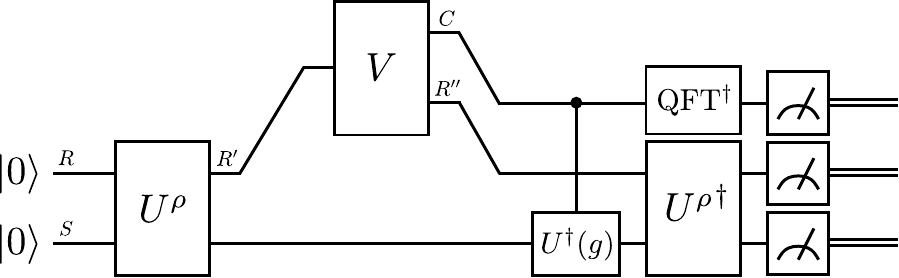}
    \caption{
    QSZK algorithm to estimate the fidelity $F(\rho, \overline{\rho})$ given a unitary $U^\rho$ that prepares a purification of $\rho$ and a unitary representation $\{U(g)\}_{g \in G}$ over which the twirled state $\overline{\rho}$ is defined. The probability of measuring the all-zeros state gives an estimate of the required fidelity. The isometry $V$ is implemented by an all-powerful prover. QFT is an abbreviation of the standard quantum Fourier transform, which in this case takes $\ket{0}$ to $\vert G \vert^{-1/2} \sum_{g \in G} \ket{g}$.}
    \label{fig:QSZK_MaxFid}
\end{figure*}

\begin{proof}
Before we get into the proof, let us recall the sine distance of two quantum states $\rho, \sigma$  \cite{R06}:
\begin{equation}
    \label{eq:sine-distance}
    P(\rho, \sigma) \coloneqq \sqrt{1-F(\rho, \sigma)}.
\end{equation}
The sine distance has a triangle-inequality property for states $\rho,\sigma,\omega$:
\begin{equation}
    \label{eq:triangle-sine}
    P(\rho, \sigma) \leq P(\rho, \omega) + P(\sigma, \omega).
\end{equation}
Furthermore, it has a data-processing inequality inherited from that of fidelity:
\begin{equation}
    P(\rho, \sigma) \geq P(\mathcal{N}(\rho), \mathcal{N}(\sigma)).
    \label{eq:DPI-sine-distance}
\end{equation}

To prove that State-$G$-Sym-Fid is QSZK-Complete, we need to show two results. First, we need to show that the problem belongs to QSZK and, second, that the problem is QSZK-Hard. 

We first show that the problem belongs to QSZK. To this end, we propose an algorithm to estimate the quantity $F(\rho, \overline{\rho})$. The underlying principle is Uhlmann's theorem \cite{Uhl76}, which can be simply understood as follows:
\begin{equation}
    F(\rho_S, \sigma_S) = \max_{V_{R \to R'}} \left\vert \langle \psi^\sigma \vert_{R'S} (V_{R \to R'} \otimes I_S) \vert \psi^\rho \rangle_{RS} \right\vert^2,
\end{equation}
where the maximization is over every isometry $V_{R \to R'}$ and $\vert \psi^\rho \rangle_{RS}$ and $\vert \psi^\sigma \rangle_{RS}$ are purifications of $\rho_S$ and $\sigma_S$, respectively. In other words, the fidelity of two states is given by the maximum squared overlap between their purifications. To calculate the fidelity of $\rho$ and $\overline{\rho}$, we then need purifications of both states. We are given a unitary $U^\rho$ to prepare a purification of $\rho$ that is used in the following manner:
\begin{align}
    \ket{\psi^\rho}_{RS} &= U^\rho_{RS} \ket{0}_{RS} , \notag \\
    \rho_S &= \Tr_{R} [\outerproj{\psi^\rho}_{RS}].
\end{align}
The following unitary generates a purification of $\overline{\rho}$:
\begin{equation}
    U^{\overline{\rho}}_{CRS} \coloneqq 
\left(\sum_g \outerproj{g}_C \otimes U_S(g)\right)  \operatorname{QFT}_C U^\rho_{RS},
\end{equation}
as follows:
\begin{align}
    \ket{\psi^{\overline{\rho}}} &= U^{\overline{\rho}}_{CR''S} \ket{0}_{CR''S} \notag\\
    &= \left(\sum_{g\in G} \outerproj{g}_C \otimes U_S(g) \right) \frac{1}{\sqrt{\vert G \vert}} \sum_{g'\in G } \ket{g'}_C  \ket{\psi^\rho}_{R''S} \notag\\
    &= \frac{1}{\sqrt{\vert G \vert}} \sum_{g\in G} \ket{g}_C \otimes U_S(g) \ket{\psi^\rho}_{R''S}.
\end{align}
Thus, performing the partial trace over $R''C$, we see that
\begin{align}
    \Tr_{R''C} [\outerproj{\psi^{\overline{\rho}}}_{R''C}] &= \frac{1}{\vert G \vert} \sum_g U_S(g) \rho_S U^\dagger_S(g) \notag\\
    &= \overline{\rho}.
\end{align}
Therefore, using the unitaries $U^\rho_{RS}$ and 
$U^{\overline{\rho}}_{R'S}$ (where $R' \equiv R''C)$, we can apply Uhlmann's theorem in tandem with an all-powerful prover (to implement the isometry~$V$) to estimate the fidelity. The construction for the algorithm can be seen in Figure~\ref{fig:QSZK_MaxFid}.

In the case of a Yes-instance,
\begin{equation}
    \maxSymG{\rho}{\sigma} \geq \alpha.
\end{equation}
Let $\sigma^\ast$ be a state achieving the maximum, i.e., 
\begin{align}
    F(\rho, \sigma^\ast) &\geq \alpha, \notag \\
    \Leftrightarrow \qquad P(\rho, \sigma^\ast) &\leq \sqrt{1-\alpha}.
\end{align}
Then
\begin{align}
    P(\rho, \overline{\rho}) &\leq P(\rho, \sigma^\ast) + P(\sigma^\ast, \overline{\rho}) \notag\\
    &\leq 2 P(\rho, \sigma^\ast)\notag\\
    &\leq 2\sqrt{1-\alpha},\notag\\
    \Leftrightarrow \qquad F(\rho, \overline{\rho}) &\geq 4\alpha - 3,
\end{align}
where the second inequality follows from the data-processing inequality (see \eqref{eq:DPI-sine-distance}) under the application of the twirling channel $\mathcal{T}_G$ and the fact that $\sigma^\ast$ is unchanged under the application of this channel.

In the case of a No-instance,
\begin{equation}
    F(\rho, \overline{\rho}) \leq \maxSymG{\rho}{\sigma} \leq \beta.
\end{equation}
Thus, there exists a gap as long as 
\begin{equation}
    \alpha > \frac{3+\beta}{4}.
\end{equation}
The interactive proof system is quantum statistical zero-knowledge because, in the case of a Yes instance, the input state $\rho$ is close to the twirled state $\overline{\rho}$. Thus, the verifier can efficiently simulate the whole interaction on their own,  and the statistical difference between the simulation and the actual protocol is negligible. As such, the problem is in the QSZK class.

Next, we show that the problem is QSZK-Hard. To do this, we map the QSZK-Complete problem State-$G$-Sym-TD to our problem. To show this, we make use of two standard inequalities that relate the trace distance and fidelity of two states \cite{fuchs1999cryptographic}:
\begin{align}
    1-\sqrt{F(\rho, \sigma)} \leq \td{\rho}{\sigma}, \label{eq:fid-TD-LHS}\\
    \sqrt{1-F(\rho, \sigma)} \geq \td{\rho}{\sigma}. \label{eq:fid-TD-RHS}
\end{align}

Consider a No-instance of State-$G$-Sym-TD. Then,
\begin{equation}
    \frac{1}{2} \min\limits_{\sigma \in \text{Sym}_G} \left\Vert \rho - \sigma \right\Vert_1 \leq \beta.
\end{equation}
Using \eqref{eq:fid-TD-LHS}, we see that
\begin{equation}
    \min\limits_{\sigma \in \text{Sym}_G} 1-\sqrt{F(\rho, \sigma)} \leq \min\limits_{\sigma \in \text{Sym}_G} \td{\rho}{\sigma} \leq \beta.
\end{equation}
Therefore, after some basic algebra,
\begin{equation}
    \max\limits_{\sigma \in \text{Sym}_G} F(\rho, \sigma) \geq (1-\beta)^2.
\end{equation}

Similarly, consider a Yes-instance of State-$G$-Sym-TD. Then,
\begin{equation}
    \min\limits_{\sigma \in \text{Sym}_G} \td{\rho}{\sigma} \geq \alpha.
\end{equation}
Using \eqref{eq:fid-TD-RHS}, we see that
\begin{equation}
    \alpha \leq \min\limits_{\sigma \in \text{Sym}_G} \td{\rho}{\sigma} \leq \min\limits_{\sigma \in \text{Sym}_G} \sqrt{1-F(\rho, \sigma)}.
\end{equation}
Therefore, after some basic algebra,
\begin{equation}
    \max\limits_{\sigma \in \text{Sym}_G} F(\rho, \sigma) \leq 1-\alpha^2.
\end{equation}

Thus, $(\alpha, \beta)$-State-$G$-Sym-TD reduces to $((1-\beta)^2, 1-\alpha^2)$-State-$G$-Sym-Fid. Since we mapped Yes (No) instances to No (Yes) instances, this proves that State-$G$-Sym-Fid belongs to co-QSZK. Since QSZK is closed under complement, the problem belongs to QSZK \cite{watrous2002qszk}. Thus, State-$G$-Sym-Fid is QSZK-Hard.
\end{proof}

\subsection{Testing \texorpdfstring{$G$}{G}-Bose Symmetric Extendibility of a State is QIP(2)-Complete}
\label{sec:GBSE-QIP2}

In this section, we show that testing $G$-Bose symmetric extendibility (G-BSE) of a state is QIP(2)-Complete. 
%\sr{To our knowledge, the proposed problem is the first non-trivial problem, other than Close Image, that is QIP(2)-Complete, addressing an open problem posed in \cite{JUW09}, wherein the authors called QIP(2) the ``most mysterious" quantum complexity class.} \sr{Need to deal with this.}

\begin{problem}
[$\left( \alpha,\beta\right)  $-State-$G$-BSE]Let $\alpha$ and $\beta$ be such that $0\leq\beta<\alpha\leq1$. Given are a circuit description of a unitary $U_{RS}^{\rho}$ that generates a purification of a state~$\rho_S$ and circuit descriptions of a unitary representation $\{U_{RS}(g)\}_{g \in G}$ of a group $G$. Decide which of the following holds:
\begin{align}
\textit{Yes: } & \quad \max\limits_{\sigma_S \in \operatorname{BSE}_G} F(\rho_S, \sigma_S) \geq \alpha,
\label{eq:G-BSE-sym-meas-yes}\\
\textit{No: } & \quad \max\limits_{\sigma_S \in \operatorname{BSE}_G} F(\rho_S, \sigma_S) \leq \beta,
\end{align}
where the set $\operatorname{BSE}_{G}$ is defined as
\begin{multline}
\operatorname{BSE}_{G}\coloneqq\\
\left\{\begin{array}[c]{c} \sigma_{S}:\exists\ \omega_{RS}\in\mathcal{D}(\mathcal{H}_{RS}), \operatorname{Tr}_R[\omega_{RS}]=\sigma_S, \\ 
\omega_{RS}=U_{RS}(g)\omega_{RS},\ \forall g\in G
\end{array}
\right\}  .
\label{eq:G-BSE-states-set}
\end{multline}
\end{problem}

Let us observe that the symmetry measure in \eqref{eq:G-BSE-sym-meas-yes} is faithful, in the sense that it is equal to one if and only if the state $\rho$ is $G$-Bose symmetric extendible. This follows from the faithfulness and continuity properties of fidelity.

\begin{theorem}
\label{thm:QIP2-GBSE} The promise problem State-$G$-BSE is QIP(2)-Complete.
\begin{enumerate}
\item $\left( \alpha,\beta\right) $-State-$G$-BSE is in QIP(2) for all $\beta<\alpha$. (It is implicit that the gap between $\alpha$ and $\beta$ is larger than an inverse polynomial in the input length.)

\item $\left( 1-\varepsilon, \varepsilon\right) $-State-$G$-BSE is QIP(2)-Hard, even when $\varepsilon$ decays exponentially in the input length.
\end{enumerate}

\noindent Thus, $\left( \alpha,\beta\right) $-State-$G$-BSE is QIP(2)-Complete for all $\left(  \alpha,\beta\right)  $ such that $0\leq \beta<\alpha\leq 1$.
\end{theorem}

\begin{proof}
To show that the problem is QIP(2)-Complete, we need to demonstrate two facts: first, that the problem is in QIP(2), and second, that it is QIP(2)-Hard. Let us begin by proving that the problem is in QIP(2). In our previous work \cite[Algorithm 3]{LRW22}, we proposed an algorithm to test for $G$-Bose symmetric extendibility of a state $\rho_S$ given a circuit description of unitary that generates a purification of the state and circuit descriptions of a unitary representation $\{U_{RS}(g)\}_{g \in G}$ of a group $G$ (see also \cite[Figure~6]{LRW22}). By inspection, the algorithm can be conducted efficiently given two messages exchanged with an all-powerful prover; therefore, this promise problem is clearly in QIP(2).

To show that the problem is QIP(2)-Hard, we map an arbitrary QIP(2) problem to a $G$-BSE problem. Specifically, from the circuit descriptions for a QIP(2) algorithm, we will identify a state $\rho_{S'}$ and a unitary representation $\{V_{R'S'}(g)\}_{g\in G}$ corresponding to a $G$-BSE problem, and we will show how the symmetry-testing condition
\begin{equation}
    \max\limits_{\sigma_{S'} \in \operatorname{BSE}_G} F(\rho_{S'}, \sigma_{S'})
\end{equation} 
can be written in terms of the QIP(2) algorithm's acceptance probability.

To begin, recall that a QIP(2) problem consists of a first verifier circuit $U^1_{SA \to S'R}$, a prover circuit $P_{RE \to R'E'}$, and a second verifier circuit $U^2_{S'R' \to DG}$, where $D$ is the decision qubit. (Here and in what follows, we keep implicit the dependence of the prover's unitary on the problem input $x$.)
The acceptance probability is given by
\begin{multline}
    p_{\operatorname{acc}} = \max\limits_{P_{RE \to R'E'}} \big\Vert (\bra{1}_D \otimes I_{E'G}) U^2_{S'R' \to DG} \times \\ P_{RE \to RE'} U^1_{SA \to S'R} \ket{x}_{S} \ket{0}_A \ket{0}_E \big\Vert^2_2.
\end{multline}
Define $\ket{\psi}_{S'R} \coloneqq  U^1_{SA \to S'R} \ket{x}_S \ket{0}_A$, and we identify the aforementioned state $\rho_{S'}$ as
\begin{equation}
    \rho_{S'} \coloneqq \Tr_{R} [\outerproj{\psi}_{S'R}].
\end{equation}
The acceptance probability can then be rewritten as
\begin{multline}
    p_{\operatorname{acc}} = \max\limits_{P_{RE \to R'E'}} \operatorname{Tr} [(\outerproj{1}_D \otimes I_{E'G}) U^2  \times \\
    P(\outerproj{0}_{E} \otimes \outerproj{\psi}_{S'R}) P^\dagger (U^2)^\dagger],
\end{multline}
where we omitted system labels for brevity.
Using the cyclicity of trace, consider that
\begin{multline}
    \label{eq:QIP(2)-accep-prob}
    p_{\operatorname{acc}} = \max\limits_{P_{RE \to R'E'}} \operatorname{Tr} [(U^2)^\dagger (\outerproj{1}_D \otimes I_{E'G}) U^2 \times \\
    P(\outerproj{0}_{E} \otimes \outerproj{\psi}_{S'R}) P^\dagger].
\end{multline}
Motivated by this, we pick the group $G$ to be $C_2$ with unitary representation $\{I_{R'S'},V_{R'S'}\}$, where 
\begin{equation}
    V_{R'S'} \coloneqq (U^2_{R'S' \to DG})^\dagger \left(- Z_D \otimes I_G \right) U^2_{R'S' \to DG},
\end{equation}
We note that $V^2_{RS} = I_{RS}$, establishing that this is indeed a representation of $C_2$. The resulting group projection is then
\begin{align}
    \Pi_{R'S'}^G &= \frac{1}{2}(I_{R'S'} + V_{R'S'}) \notag\\ 
    &= \frac{1}{2} \left( {U^2}^{\dag}(I_D \otimes I_G)U^2 + {U^2}^{\dag} \left( -Z_D \otimes I_G \right) U^2 \right) \notag\\
    &=\frac{1}{2} \left( {U^2}^{\dag} \left( I_{DG} - \left( Z_D \otimes I_G \right) \right) U^2\right) \notag\\
    &= {U^2}^{\dag} (\outerproj{1}_D \otimes I_G)U^2,
\end{align}
which is precisely the acceptance projection in the first line of \eqref{eq:QIP(2)-accep-prob}. That is,
\begin{equation}
    p_{\operatorname{acc}} = \max\limits_{P_{RE \to R'E'}} \operatorname{Tr} [ \Pi_{R'S'}^G
    P(\outerproj{0}_{E} \otimes \outerproj{\psi}_{S'R}) P^\dagger].
\end{equation}
Now invoking \cite[Theorem III.3]{LRW22}, we conclude that
\begin{align}
    & \max\limits_{P_{RE \to R'E'}} \operatorname{Tr} [ \Pi_{R'S'}^G
    P(\outerproj{0}_{E} \otimes \outerproj{\psi}_{S'R}) P^\dagger] \notag \\
    & = \max\limits_{P_{RE \to R'E'}} \left \Vert  \Pi_{R'S'}^G
    P(\ket{0}_{E} \otimes \ket{\psi}_{S'R}) \right\Vert_2^2 \notag \\
    & = \max_{\sigma_{S'} \in \operatorname{BSE}_G} F(\rho_{S'}, \sigma_{S'}),
\end{align}
where $\operatorname{BSE}_{G}$ in this case is
\begin{multline}
\operatorname{BSE}_{G}\coloneqq\\
\left\{\begin{array}[c]{c} \sigma_{S'}:\exists\ \omega_{R'S'}\in\mathcal{D}(\mathcal{H}_{R'S'}), \operatorname{Tr}_{R'}[\omega_{R'S'}]=\sigma_{S'}, \\ 
\omega_{R'S'}=V_{R'S'}\omega_{R'S'},
\end{array}
\right\}  .
\label{eq:G-BSE-states-set-QIP2}
\end{multline}
To help visualize the reduction,
the $G$-Bose symmetric extendibility test corresponding to a general QIP(2) algorithm is depicted in Figure~\ref{fig:QIP(2)-GBSE}.
    
\begin{figure*}
    \centering
    \includegraphics[width=0.9\linewidth]{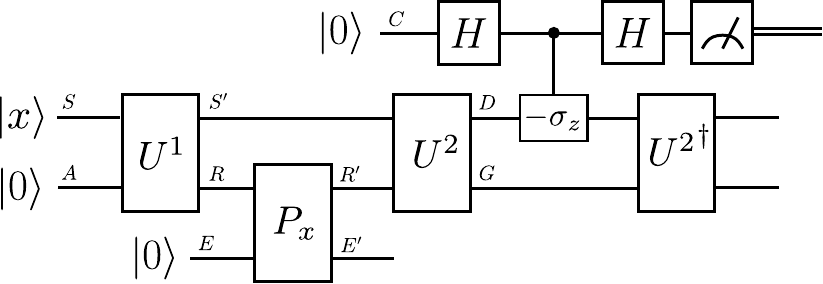}
    \caption{Circuit to map an arbitrary QIP(2) computation to a $G$-Bose symmetric extendibility test.}
\label{fig:QIP(2)-GBSE}
\end{figure*}

Thus, the acceptance probability of the QIP(2) algorithm exactly matches the symmetry-testing condition of the constructed $G$-BSE problem. As such, any QIP(2) problem can be efficiently mapped to a $G$-BSE problem, proving that the problem State-$G$-BSE is QIP(2)-Hard. Along with the fact that the problem lies in the QIP(2) class, this concludes the proof.
\end{proof}

\subsection{Testing \texorpdfstring{$G$}{G}-Bose Symmetric Separable Extendibility is \texorpdfstring{\qipeb}{QIP-EB(2)}-Complete}

\label{sec:GBS-QIPEB}

In this section, we introduce the following problem: decide whether a state has a separable extension that is $G$-Bose symmetric. We also prove that it is a  \qipeb-Complete problem. 

\begin{problem}
[$\left(\alpha,\beta\right)$-Sep-Ext-$G$-Bose-Symmetry]Let $\alpha$ and $\beta$ be such that $0\leq \beta < \alpha \leq1$. Given is a circuit description of a unitary $U_{SS'}^{\rho}$ that generates a purification of a state $\rho_S$, as well as circuit descriptions of a unitary representation $\{U_{RS}(g)\}_{g \in G}$ of a group~$G$. Decide which of
the following holds:
\begin{align}
\textit{Yes: } & \quad 
\max_{\substack{\omega_{RS} \in \operatorname{SEP}(R:S), \\ \Tr_R[\omega_{RS}] = \rho_S}} \Tr[\Pi^G_{RS} \omega_{RS}]
 \geq \alpha, 
\label{eq:test-G-bose-sep-sym-ext}
\\
\textit{No: }  & \quad \max_{\substack{\omega_{RS} \in \operatorname{SEP}(R:S), \\ \Tr_R[\omega_{RS}] = \rho_S}} \Tr[\Pi^G_{RS} \omega_{RS}]
 \leq \beta,
\end{align}
where $\Pi^G_{RS}$ is defined in \eqref{eq:Pi-proj-RS-def}.
\end{problem}

Let us observe that the following equality holds:
\begin{multline}
    \max_{\substack{\omega_{RS} \in \operatorname{SEP}(R:S), \\ \Tr_R[\omega_{RS}] = \rho_S}} \Tr[\Pi^G_{RS} \omega_{RS}] = \\
    \max_{\substack{\omega_{RS} \in \operatorname{SEP}(R:S), \\ \Tr_R[\omega_{RS}] = \rho_S, \\ \sigma_{RS} \in \text{B-Sym}_G}} F(\omega_{RS},\sigma_{RS}),
    \label{eq:fid-qip-eb-2-rewrite}
\end{multline}
where the set $\operatorname{B-Sym}_{G}$ in this case is defined as
\begin{equation}
\operatorname{B-Sym}_{G} \coloneqq
\left\{ \sigma_{RS}: \sigma_{RS} = U_{RS}(g)\sigma_{RS},\ \forall g\in G \right\}.
\label{eq:G-BSSE-states-set}
\end{equation}
The identity in \eqref{eq:fid-qip-eb-2-rewrite} follows as a consequence of \cite[Theorem~3.1]{LRW22}. Rewriting the expression in this way allows for a fidelity interpretation of the symmetry condition, which implies that the symmetry-testing condition in \eqref{eq:test-G-bose-sep-sym-ext} is  equal to one if and only if there exists a separable extension of $\rho_S$  that is Bose-symmetric according to the unitary representation $\{U_{RS}(g)\}_{g \in G}$. As such, this is a faithful measure of $G$-Bose symmetric separable extendibility.

\begin{theorem}
\label{thm:QIPEB2-Sep-Ext-GBS} The promise problem Sep-Ext-$G$-Bose-Symmetry is \qipeb-Complete.

\begin{enumerate}
\item $\left(  \alpha,\beta\right)  $-Sep-Ext-$G$-Bose-Symmetry is in \qipeb\ for all $\beta<\alpha$. (It is implicit that the gap between $\alpha$ and $\beta$ is larger than an inverse polynomial in the input length.)

\item $\left(   1-\varepsilon, \varepsilon\right)  $-Sep-Ext-$G$-Bose-Symmetry is \qipeb-Hard, even when $\varepsilon$ decays exponentially in the input length.
\end{enumerate}

\noindent Thus, $\left(  \alpha,\beta\right)  $-Sep-Ext-$G$-Bose-Symmetry is \qipeb-Complete for all $\left(  \alpha,\beta\right)  $ such that $0\leq \beta<\alpha\leq 1$.
\end{theorem}

\begin{proof}
To show that the problem is \qipeb-Complete, we need to demonstrate two facts: first, that the problem is in \qipeb, and second, that it is \qipeb-Hard. Let us begin by proving that the problem is in \qipeb. 

\begin{figure}
    \centering
    \includegraphics[width=1.0\columnwidth]{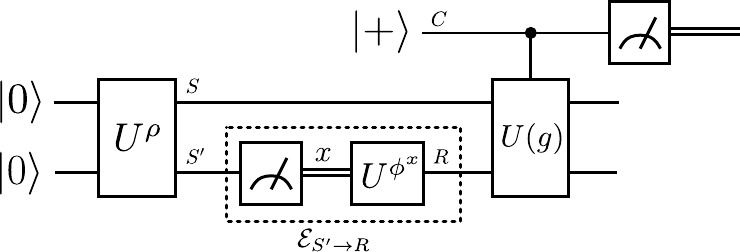}
    \caption{\qipeb\ algorithm to test for $G$-Bose symmetry of a separable extension of the state $\rho_S$, where the prover's actions are depicted in the dashed box. The prover's channel $\mathcal{E}_{S' \to R}$ is entanglement breaking.}
    \label{fig:sep-ext-GBS}
\end{figure}

The algorithm to estimate the quantity in \eqref{eq:test-G-bose-sep-sym-ext} is given in Figure~\ref{fig:sep-ext-GBS}. The general form of an entanglement-breaking channel is given by
\begin{equation}
    \mathcal{E}_{S' \to R}(\cdot) = \sum_{x \in \mathcal{X}} \Tr[\mu^x_{S'}(\cdot)]\phi^x_{R},
\end{equation}
as discussed in \eqref{eq:EB-ch-def}.
Thus, the acceptance probability of the algorithm for a fixed channel $\mathcal{E}_{S' \to R}$ is given by
\begin{align}
    &\Tr[\Pi^G_{RS} \mathcal{E}_{S' \to R}(\outerproj{\psi^\rho}_{SS'})] \notag \\
    &= \sum_{x \in \mathcal{X}} \Tr[\Pi^G_{RS} \Tr_{S'}[\mu^x_{S'}(\outerproj{\psi^\rho}_{SS'})] \otimes \phi^x_{R}] \notag \\
    &= \sum_{x \in \mathcal{X}} p(x) \Tr[\Pi^G_{RS} (  \phi^x_{R} \otimes \psi^x_{S})],
    \label{eq:qipeb2-fixed-accep}
\end{align}
where 
\begin{align}
    p(x) &\coloneqq \Tr[\mu^x_{S'}(\outerproj{\psi^\rho}_{SS'})], \\
    \psi^x_{S} &\coloneqq \frac{1}{p(x)} \Tr_{S'}[\mu^x_{S'}(\outerproj{\psi^\rho}_{SS'})].
\end{align}
Note that
\begin{align}
    \sum_{x \in \mathcal{X}} p(x) \psi^x_S 
    &= \sum_{x \in \mathcal{X}} \Tr_{S'}[\mu^x_{S'}(\outerproj{\psi^\rho}_{SS'})] \notag \\
    &= \Tr_{S'}[\outerproj{\psi^\rho}_{SS'}] \notag \\
    &= \rho_S.
\end{align}

Thus, maximizing over all possible entanglement-breaking channels, the acceptance probability is given by
\begin{multline}
    \max_{\{(p(x), \psi^x_S, \phi^x_{R})\}_x} \sum_{x \in \mathcal{X}} p(x) \Tr[\Pi^G_{RS} (\phi^x_{R} \otimes \psi^x_{S})]  \\
    = \max_{\omega_{RS} \in \operatorname{SEP}(R:S)} \Tr[\Pi^G_{RS} \omega_{RS}],
\end{multline}
with the condition that $\Tr_R[\omega_{RS}] = \sum_x p(x) \psi^x_S = \rho_S$.
Since the entire algorithm can be performed efficiently when augmented by an entanglement-breaking prover, the problem is in \qipeb.

Next, we show that the problem is \qipeb-Hard. To do this, we need to map a general \qipeb\ problem to an instance of Sep-Ext-$G$-Bose-Symmetry; i.e., using the circuit descriptions for a general \qipeb\ algorithm, we need to define a state~$\rho_{S}$ and a unitary representation $\{U_{RS}(g)\}_{g \in G}$.  

Consider a general interactive proof system in \qipeb\ that begins with the verifier preparing a bipartite pure state $\psi_{RS}$, followed by the system $R$ being sent to the prover, who subsequently performs an entanglement-breaking channel $\mathcal{E}_{R \to R'}$ and sends the $R'$ register to the verifier. The verifier then performs a unitary $V_{R^{\prime}S\rightarrow DG}$, measures the decision qubit, and accepts if the outcome $\ket{1}$ is observed. Indeed, the acceptance probability is given by
\begin{equation}
    \max_{\mathcal{E} \in \operatorname{EB}} \Tr[ (\outerproj{1}_{D} \otimes I_{G}) \mathcal{V}_{R'S \rightarrow DG} (\mathcal{E}_{R \rightarrow R'}(\psi_{RS}))],
\end{equation}
where $\mathcal{V}_{R'S \rightarrow DG}$ is the unitary channel corresponding to the unitary operator $V_{R'S \rightarrow DG}$ and EB denotes the set of entanglement-breaking channels. Following the reasoning in \eqref{eq:qipeb2-fixed-accep}, the output of an entanglement-breaking channel can be written in the form
\begin{equation}
    \mathcal{E}_{R\rightarrow R^{\prime}}(\psi_{RS})= \sum_{x} p(x) \phi_{R'}^x \otimes \psi_S^x,
\end{equation}
with the condition that $\sum_x p(x) \psi^x_S = \Tr_R[\psi_{RS}] = \rho_S$. In other words, the output of the entanglement-breaking channel is a separable extension of the state $\rho_S$. Thus, the acceptance probability is given by
\begin{multline}
    \max_{\omega_{R'S} \in \operatorname{SEP}(R':S)} \Tr[ (\outerproj{1}_{D} \otimes I_{G}) V (\omega_{R'S}) V^\dagger]  \\
    = \max_{\omega_{R'S} \in \operatorname{SEP}(R':S)} \Tr[( V^\dagger (\outerproj{1}_{D} \otimes I_{G}) V )\omega_{R'S}],
\end{multline}
subject to the constraint $\Tr_{R'}[\omega_{R'S}] = \rho_S$, where we have used the shorthand $V \equiv V_{R'S \to DG}$. The second inequality results from the cyclicity of trace. 

Let us then define the group $G$ to be the cyclic group on two elements $C_2 = \{I, W\}$ such that $W^2 = I$, where $W$ is simply given by
\begin{equation}
    W_{R'S} \coloneqq (V_{RS' \to DG})^\dagger \left(- Z_D \otimes I_G \right) V_{R'S \to DG},
\end{equation}
We note that $W^2_{R'S} = I_{R'S}$, establishing that this is indeed a representation of $C_2$. The resulting group projection is then
\begin{align}
    \Pi_{R'S}^G &= \frac{1}{2}(I_{R'S} + W_{R'S}) \notag\\ 
    &= \frac{1}{2} \left( V^\dagger(I_D \otimes I_G)V + V^\dagger \left( -Z_D \otimes I_G \right) V \right) \notag\\
    &=\frac{1}{2} \left( V^\dagger \left( I_{DG} - \left( Z_D \otimes I_G \right) \right) V \right) \notag\\
    &= V^\dagger (\outerproj{1}_D \otimes I_G) V.
\end{align}
Next, we define the state $\rho_S$ to be
\begin{equation}
    \rho_S \coloneqq \Tr_R[\psi_{RS}].
\end{equation}

Thus, the symmetry-testing condition of the instance of Sep-Ext-$G$-Bose-Symmetry is given by
\begin{multline}
     \max_{\omega_{R'S} \in \operatorname{SEP}(R':S)} \Tr[\Pi^G_{R'S} \omega_{R'S}]  = \\
     \max_{\omega_{R'S} \in \operatorname{SEP}(R':S)} \Tr[ V^\dagger (\outerproj{1}_D \otimes I_G) V \omega_{R'S}],
\end{multline}
where the maximization over $\omega_{R'S}$ is subject to the constraint that $\Tr_{R'}[\omega_{R'S}] = \rho_S$. 
This exactly matches the acceptance probability of the \qipeb\ problem, establishing that Sep-Ext-$G$-Bose-Symmetry is \qipeb-Hard, thus completing the proof. 
\end{proof}

\subsection{Testing \texorpdfstring{$G$}{G}-Bose Symmetric Extendibility of the Output of a Channel is QIP-Complete}
\label{sec:GBSE-QIP}

In this section, we show that testing the $G$-Bose symmetric extendibility (G-BSE) of the output of a channel state is QIP-Complete.

\begin{problem}
[$\left( \alpha,\beta\right)  $-Channel-$G$-BSE]Let $\alpha$ and $\beta$ be such that $0\leq\beta < \alpha \leq1$. Given are descriptions of circuits $U_{BC' \to S'S}^{\mathcal{N}}$ that prepare a unitary dilation of a channel
\begin{multline}
\mathcal{N}_{B \to S} (\cdot) \coloneqq \\
\operatorname{Tr}_{S'}[U^\mathcal{N}_{BC' \to S'S}((\cdot)_B \otimes \outerproj{0}_{C'}) (U^\mathcal{N}_{BC' \to S'S})^\dagger]    
\end{multline}
and descriptions of a unitary representation $\{U_S(g)\}_{g \in G}$ of a group $G$. Decide which of
the following holds:
\begin{align}
\textit{Yes: } & \quad  \max\limits_{\substack{\rho_B \in \mathcal{D}(\mathcal{H}_B), \\\sigma_S \in \operatorname{BSE}_G}} F(\mathcal{N}_{B \to S}(\rho_B), \sigma_S) \geq \alpha,
\label{eq:ch-G-BSE-sym-meas-yes}\\
\textit{No: } & \quad \max\limits_{\substack{\rho_B \in \mathcal{D}(\mathcal{H}_B), \\\sigma_S \in \operatorname{BSE}_G}} F(\mathcal{N}_{B \to S}(\rho_B), \sigma_S) \leq \beta,
\end{align}
where the set $\operatorname{BSE}_{G}$ is defined to be:
\begin{multline}
\operatorname{BSE}_{G}\coloneqq\\
\left\{\begin{array}[c]{c} \sigma_{S}:\exists\ \omega_{RS}\in\mathcal{D}(\mathcal{H}_{RS}), \operatorname{Tr}_R[\omega_{RS}]=\sigma_S, \\ 
\omega_{RS}=U_{RS}(g)\omega_{RS},\ \forall g\in G
\end{array}
\right\}.
\end{multline}
\end{problem}

Let us observe that the symmetry measure in \eqref{eq:ch-G-BSE-sym-meas-yes} is faithful, in the sense that it is equal to one if and only if there is a channel input state $\rho_B$ such that the output state $\mathcal{N}_{B \to S}(\rho_B)$ is $G$-Bose symmetric extendible.

\begin{theorem}
\label{thm:ChlOut-QIP-GBSE} The promise problem Channel-$G$-BSE is QIP-Complete. 

\begin{enumerate}
\item $\left( \alpha,\beta\right) $-Channel-$G$-BSE is in QIP for all $\beta<\alpha$. (It is implicit that the gap between $\alpha$ and $\beta$ is larger than an inverse polynomial in the input length.)

\item $\left( 1-\varepsilon, \varepsilon\right) $-Channel-$G$-BSE is QIP-Hard, even when $\varepsilon$ decays exponentially in the input length.
\end{enumerate}

\noindent Thus, $\left( \alpha,\beta\right) $-Channel-$G$-BSE is QIP-Complete for all $\left(  \alpha,\beta\right)  $ such that $0\leq \beta<\alpha\leq 1$.
\end{theorem}

\begin{proof}
To show that the problem is QIP-Complete, we need to demonstrate two facts: first, that the problem is in QIP, and second, that it is QIP-Hard. Let us begin by proving that the problem is in QIP. In our previous work \cite[Algorithm 3]{LRW22}, we proposed an algorithm to test for $G$-Bose Symmetric Extendibility of a state $\rho_S$ given a circuit description of unitary that generates a purification of the state and circuit descriptions of a unitary  representation $\{U_{RS}(g)\}_{g \in G}$ of a group $G$. By inspection, the algorithm can be executed efficiently given two messages exchanged with an all-powerful prover. The optimal input state to the channel is sent by another message of the prover, thus adding up to three messages in total. As such, the algorithm is clearly in QIP.

To show that the problem is QIP-Hard, we map an arbitrary QIP problem to an instance $(\mathcal{N}, \{U_{RS}(g)\}_{g \in G})$ of Channel-$G$-BSE. Since $\text{QIP(3)} \equiv \text{QIP}$ \cite{kitaev2000parallelization}, our goal is to find a correspondence between an arbitrary QIP(3) protocol and a choice of channel and group, $(\mathcal{N}, \{U_{RS}(g)\}_{g \in G})$. Specifically, from the circuit descriptions for a QIP(3) algorithm, we will identify a channel $\mathcal{N}_{R''\to S'}$ and a unitary representation $\{V_{R'S'}(g)\}_{g\in G}$ corresponding to a $G$-BSE problem, and we will show how the symmetry-testing condition
\begin{equation}
    \max\limits_{\substack{\rho_{R''} \in \mathcal{D}(\mathcal{H}_{R''}), \\\sigma_{S'} \in \operatorname{BSE}_G}} F(\mathcal{N}_{R'' \to S'}(\rho_{R''}), \sigma_{S'})
\end{equation}
can be written in terms of the QIP(3) algorithm's acceptance probability.

To begin, recall that an arbitrary QIP(3) problem consists of three messages exchanged and involves a first prover unitary $P^1_{E'' \to R''E}$, a first verifier unitary $U^1_{SAR'' \to S'R}$, a second prover unitary $P^2_{RE \to R'E'}$, and a second verifier unitary $U^2_{S'R' \to DG}$, where $D$ is the decision qubit. (Here we leave the dependence of the prover unitaries on $x$ to be implicit.)
The acceptance probability is thus,
\begin{multline}
    p_{\operatorname{acc}} = \max\limits_{\substack{P^1_{E'' \to R''E}, \\
    P^2_{RE \to R'E'}}} \big\Vert (\bra{1}_D \otimes I_{GE'}) U^2_{S'R' \to DG} P^2_{RE \to R'E'} \\ 
    \times  U^1_{SAR'' \to S'R} P^1_{E'' \to R''E} \ket{x}_{S} \ket{0}_{A} \ket{0}_{E''} \big\Vert^2_2.
\end{multline}
Defining the first state after the action of the prover's unitary $P^1$ to be
\begin{equation}
    \ket{\psi}_{R''E} \coloneqq P^1_{E'' \to R''E} \ket{0}_{E''} ,
\end{equation}
and the isometry
\begin{equation}
    W_{R'' \to S'R} \coloneqq U^1_{SAR'' \to S'R}  \ket{x}_{S} \ket{0}_{A},
\end{equation}
the acceptance probability can then be written as 
\begin{align}
    p_{\operatorname{acc}} & = \max\limits_{\psi_{R''E},P^2_{RE \to R'E'}} \operatorname{Tr} [(\outerproj{1}_D \otimes I_{E'G}) U^2  \notag \\ 
    & \qquad \times P^2 W \psi_{R''E}  {W}^\dagger {P^2}^\dagger {U^2}^\dagger ]\notag \\
    & = \max\limits_{\psi_{R''E},P^2_{RE \to R'E'}} \operatorname{Tr} [{U^2}^\dagger(\outerproj{1}_D \otimes I_{E'G}) U^2  \notag \\
    & \qquad \times P^2 W \psi_{R''E}  {W}^\dagger {P^2}^\dagger  ] ,
    \label{eq:QIP-accep-prob}
\end{align}
where we have used cyclicity of trace in the last line. We can also identify the aforementioned channel $\mathcal{N}_{R''\to S'}$  as follows:
\begin{equation}
    \mathcal{N}_{R''\to S'}(\cdot) \coloneqq 
    \Tr_{R} [W (\cdot)_{R''}  W^\dagger].
\end{equation}

Motivated by the expression in \eqref{eq:QIP-accep-prob}, we pick the group~$G$ to be $C_2$ with unitary representation $\{I_{R'S'},V_{R'S'}\}$, where 
\begin{equation}
    V_{R'S'} \coloneqq (U^2_{R'S' \to DG})^\dagger \left(- Z_D \otimes I_G \right) U^2_{R'S' \to DG}.
\end{equation}
We note that $V^2_{RS} = I_{RS}$, proving that this is indeed a representation of $C_2$. The resulting group projection is then
\begin{align}
    \Pi_{R'S'}^G &= \frac{1}{2}(I_{R'S'} + V_{R'S'}) \notag\\ 
    &= \frac{1}{2} \left( {U^2}^{\dag}(I_D \otimes I_G)U^2 + {U^2}^{\dag} \left( -Z_D \otimes I_G \right) U^2 \right) \notag\\
    &=\frac{1}{2} \left( {U^2}^{\dag} \left( I_{DG} - \left( Z_D \otimes I_G \right) \right) U^2\right) \notag\\
    &= {U^2}^{\dag} (\outerproj{1}_D \otimes I_G)U^2,
\end{align}
which is precisely the acceptance projection in \eqref{eq:QIP-accep-prob}. That is,
\begin{equation}
    p_{\operatorname{acc}} = \max\limits_{\substack{\psi_{R''E}, \\ P^2_{RE \to R'E'}}} \operatorname{Tr} [\Pi_{R'S'}^G P^2 W \psi_{R''E}  {W}^\dagger {P^2}^\dagger  ]
\end{equation}
Now invoking \cite[Theorem III.3]{LRW22}, we conclude that
\begin{align}
    p_{\operatorname{acc}} & = \max\limits_{\psi_{R''E},P^2_{RE \to R'E'}} \operatorname{Tr} [\Pi_{R'S'}^G P^2 W \psi_{R''E}  {W}^\dagger {P^2}^\dagger  ] \notag \\
    & = \max\limits_{\psi_{R''E},P^2_{RE \to R'E'}} \left\Vert \Pi_{R'S'}^G P^2 W \ket{\psi}_{R''E}    \right\Vert_2^2 \notag \\
    &  = \max\limits_{\substack{\rho_{R''} \in \mathcal{D}(\mathcal{H}_{R''}), \\\sigma_{S'} \in \operatorname{BSE}_G}} F(\mathcal{N}_{R'' \to S'}(\rho_{R''}), \sigma_{S'}),
\end{align}
where in this case $ \operatorname{BSE}_G$ is defined in the same way as in~\eqref{eq:G-BSE-states-set-QIP2}. To help visualize the reduction, 
the $G$-Bose symmetric extendibility test corresponding to a general QIP algorithm is depicted in Figure~\ref{fig:QIP-Chlout-GBSE}.

\begin{figure*}
    \centering
    \includegraphics[width=0.9\linewidth]{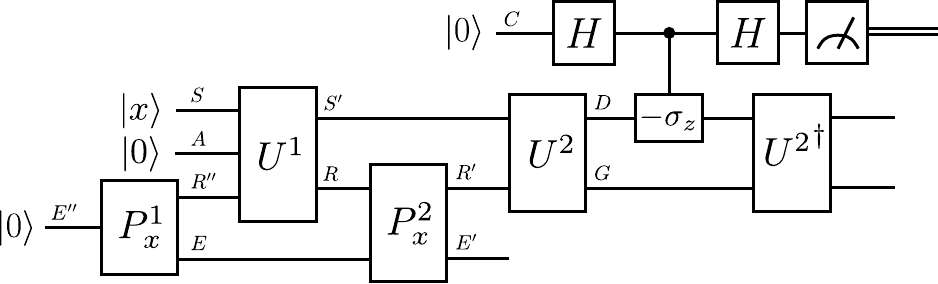}
    \caption{Circuit to map an arbitrary QIP algorithm to a $G$-Bose symmetric extendibility test on the output of a channel.}
\label{fig:QIP-Chlout-GBSE}
\end{figure*}

Thus, the acceptance probability of the QIP algorithm now exactly matches the symmetry-testing condition of the constructed $G$-BSE problem. As such, any QIP problem can be efficiently mapped to testing $G$-BSE of the output of a channel, proving that the problem Channel-$G$-BSE is QIP-Hard. Along with the fact that the problem lies in the QIP class, this concludes the proof.
\end{proof}

\subsection{Testing Hamiltonian Symmetry Using Maximum Spectral Norm is in QMA}
\label{sec:HS-QMA}

In this section, we show that testing whether a Hamiltonian is symmetric with respect to a group representation and the maximum spectral norm is in QMA. In particular, we consider the following task: given a group~$G$ with unitary representation $\{U(g)\}_{g\in G}$, a time $t \in \mathbb{R}$, and a classical description of a local or sparse Hamiltonian $H$, estimate the following quantity:
\begin{equation}
    \max_{g \in G} \left\Vert [U(g), e^{-iHt}] \right\Vert^2_\infty,
    \label{eq:max-spec-norm-Ham}
\end{equation}
where the spectral norm of a matrix $A$ is defined as
\begin{equation}  
\left \Vert A \right \Vert_\infty \coloneqq \sup_{\ket{\psi} \in \mathcal{H} } \left\{\left \Vert A \ket{\psi} \right \Vert_2 : \left \Vert  \ket{\psi} \right \Vert_2 = 1\right\}.
\end{equation}
The quantity in \eqref{eq:max-spec-norm-Ham} is a faithful measure of asymmetry in the following sense:
\begin{align}
      \max_{g \in G} \left\Vert [U(g), e^{-iHt}] \right\Vert^2_\infty  & =  0 \quad \forall t \in (-\delta,\delta), \notag  \\
     \Leftrightarrow  \qquad
     [U(g), e^{-iHt}]   & =  0 \quad\forall g \in G, t \in (-\delta,\delta), \notag \\
      \Leftrightarrow \qquad
     [U(g), H]   & =  0 \quad\forall g \in G, 
\end{align}
where $\delta > 0$. The first equivalence follows from faithfulness of the spectral norm, and the second equivalence follows by taking the derivative of the second line at $t=0$.

\begin{problem}[Ham-Sym-Max-Spec]
\label{prob:QMA-HamSym-test} Let $\alpha$ and $\beta$ be such that $0\leq\beta<\alpha\leq 2$, and fix $t \in \mathbb{R}$. Given are circuit descriptions of a unitary representation $\{U(g)\}_{g \in G}$ of a group~$G$ and a classical description of a $k$-local or sparse Hamiltonian~$H$. Decide which of the following holds:
\begin{align}
\textit{Yes:} & \quad \max_{g \in G} \left\Vert [U(g), e^{-iHt}] \right\Vert^2_\infty \geq \alpha,\\
\textit{No:} & \quad \max_{g \in G} \left\Vert [U(g), e^{-iHt}] \right\Vert^2_\infty \leq \beta,
\end{align}
\end{problem}

In what follows, we show that Ham-Sym-Max-Spec is in QMA, and it remains an interesting open question to determine whether this problem is QMA-Hard or hard for some other complexity class.

\begin{theorem}
\label{thm:QMA-HamSym-test} The promise problem Ham-Sym-Max-Spec is in QMA.
\end{theorem}

\begin{proof}
\begin{figure}
    \centering
    \includegraphics[width=1.0\columnwidth]{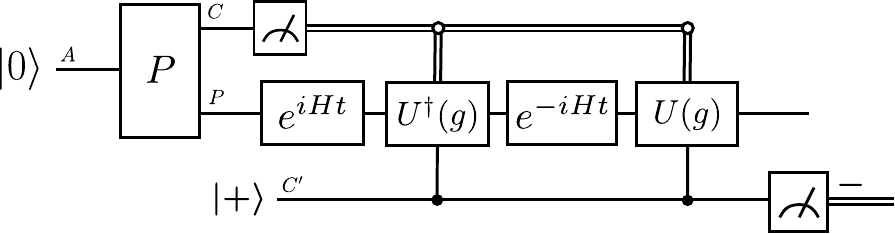}
    \caption{Circuit depicting a QMA test for Hamiltonian symmetry with respect to a group, where it is understood that the unitary $P$ is implemented by an all-powerful prover. The final measurement is in the Hadamard basis, and the algorithm accepts if the $\ket{-}$ outcome occurs.}
    \label{fig:ham-sym-max-spec}
\end{figure}
Consider the following steps of a QMA interactive proof (see Figure~\ref{fig:ham-sym-max-spec}):
\begin{enumerate}
    \item The prover sends a state in registers $C$ and $P$, with the dimension of $C$ being equal to $|G|$ and the dimension of $P$ being equal to the dimension of $H$.
    \item The verifier measures the register $C$ and obtains the outcome $g\in G$.
    \item The verifier adjoins a qubit $C'$ in the state $|+\rangle$, performs the Hamiltonian evolution $e^{iHt}$, the controlled unitary $|0\rangle\!\langle 0| \otimes I + |1\rangle\!\langle 1| \otimes U^\dag(g) $, the Hamiltonian evolution $e^{-iHt}$, and the controlled unitary $|0\rangle\!\langle 0| \otimes I + |1\rangle\!\langle 1| \otimes U(g) $.
    \item The verifier measures the qubit $C'$ in the Hadamard basis $\{|+\rangle,|-\rangle\}$ and accepts if the outcome $|-\rangle$ occurs.
\end{enumerate}
As noted in the previous section, there exist multiple methods to realize an efficient circuit for the Hamiltonian evolutions $e^{-iHt}$ and $e^{iHt}$ (see \cite{childs2018toward} and references therein). We also note that there are some similarities, as well as key differences, between this algorithm and that given in Figure~3 of \cite{pg2021exponential}.

Let us now analyze the acceptance probability of this interactive proof. It suffices for the prover to send a pure state, as this maximizes the acceptance probability. Let us expand a fixed pure state $|\psi\rangle_{CP}$ of registers $C$ and $P$ as follows:
\begin{equation}
    |\psi\rangle_{CP} = \sum_{g \in G} \sqrt{p(g)} |g\rangle_C |\psi_g\rangle_P,
\end{equation}
where $\{p(g)\}_g$ is a probability distribution and $\{|\psi_g\rangle_P\}_g$ is a set of states. After the verifier's measurement in Step~2, the probability of obtaining outcome $g\in G$ is $p(g)$ and the post-measurement state of register $P$ is $|\psi_g\rangle_P$. Conditioned on the outcome $g$ and defining the unitary $W(g,t) \equiv U(g)e^{-iHt} U^\dag(g)e^{iHt}$, the acceptance probability of Steps~3-4 is then given by
\begin{multline}
     \left\Vert (\bra{-}_{C'} \otimes I_P) \frac{1}{\sqrt{2}} (|0\rangle_{C'} |\psi_g\rangle_P +|1\rangle_{C'} W(g,t) |\psi_g\rangle_P) \right\Vert_2^2  \\
     = \frac{1}{4} \left\Vert  (I - W(g,t) )|\psi_g\rangle_P \right\Vert_2^2.
\end{multline}
Thus, for a fixed state $ |\psi\rangle_{CP}$ of the prover, the acceptance probability is given by
\begin{equation}
    \frac{1}{4} \sum_{g\in G} p(g) \left\Vert  (I - W(g,t) )|\psi_g\rangle_P \right\Vert_2^2,
\end{equation}
and finally maximizing over all such states leads to the following expression for the acceptance probability:
\begin{align}
    & \max_{|\psi\rangle_{CP}} \frac{1}{4} \sum_{g\in G} p(g) \left\Vert  (I - W(g,t) )|\psi_g\rangle_P \right\Vert_2^2 \notag \\
    & = \frac{1}{4} \max_{\substack{\{p(g)\}_g, \\ \{|\psi_g\rangle_P\}_g}}  \sum_{g\in G} p(g) \left\Vert  (I - W(g,t) )|\psi_g\rangle_P \right\Vert_2^2 \notag \\
    & = \frac{1}{4} \max_{\{p(g)\}_g  }  \sum_{g\in G} p(g) \max_{\{|\psi_g\rangle_P\}_g} \left\Vert  (I - W(g,t) )|\psi_g\rangle_P \right\Vert_2^2 \notag \\
    & = \frac{1}{4} \max_{\{p(g)\}_g  }  \sum_{g\in G} p(g)  \left\Vert  I - W(g,t)  \right\Vert_\infty^2 \notag \\
    & = \frac{1}{4} \max_{g \in G}   \left\Vert  I - W(g,t) \right\Vert_{\infty}^2 \notag \\
    & = \frac{1}{4} \max_{g \in G}   \left\Vert  I - U(g)e^{-iHt} U^\dag(g)e^{iHt} \right\Vert_{\infty}^2 \notag \\
    & = \frac{1}{4} \max_{g \in G}   \left\Vert  e^{-iHt} U(g) - U(g)e^{-iHt}  \right\Vert_{\infty}^2 \notag \\
    & = \frac{1}{4} \max_{g \in G} \left\Vert [U(g), e^{-iHt}] \right\Vert^2_\infty. \label{eq:QMA-block-eq}
\end{align}
The third equality follows from the definition of the spectral norm.
The fourth equality follows because the optimal distribution is a point mass on the largest value of $\left\Vert  I - W(g,t) \right\Vert_\infty^2$.  The penultimate equality follows from unitary invariance of the spectral norm. In light of the above analysis, the best strategy of the prover is to compute $\max_{g \in G} \left\Vert [U(g), e^{-iHt}] \right\Vert^2_\infty$ in advance, send the maximizing value of $g$ in register $C$, and send the corresponding state that achieves the spectral norm in register $P$. As the acceptance probability of this QMA interactive proof is precisely related to the decision criteria in Problem~\ref{prob:QMA-HamSym-test}, this concludes the proof.
\end{proof}

\subsection{Testing Hamiltonian Symmetry Using Average Spectral Norm is in QAM}
\label{sec:HS-QAM}

In this section, we show that testing whether a Hamiltonian is symmetric with respect to a group representation and the average spectral norm is in QAM. In particular, we consider the following task: given a group $G$ with unitary representation $\{U(g)\}_{g\in G}$, a time $t \in \mathbb{R}$, and a classical description of a local or sparse Hamiltonian $H$, estimate the following quantity:
\begin{equation}
    \frac{1}{|G|} \sum_{g \in G} \left\Vert [U(g), e^{-iHt}] \right\Vert^2_\infty. 
\end{equation}
This is a faithful measure of symmetry in the following sense:
\begin{align}
      \frac{1}{|G|} \sum_{g \in G} \left\Vert [U(g), e^{-iHt}] \right\Vert^2_\infty  & =  0 \quad \forall t \in (-\delta,\delta), \notag  \\
     \Leftrightarrow  \qquad
     [U(g), e^{-iHt}]   & =  0 \quad\forall g \in G, t \in (-\delta,\delta), \notag \\
      \Leftrightarrow \qquad
     [U(g), H]   & =  0 \quad\forall g \in G, 
\end{align}
where $\delta > 0$. The first equivalence follows from faithfulness of the spectral norm, and the second equivalence follows by taking the derivative of the second line at $t=0$.

\begin{problem}[Ham-Sym-Avg-Spec]
\label{prob:QAM-HamSym-test} Let $\alpha$ and $\beta$ be such that $0\leq\beta<\alpha\leq \gamma$, where $\gamma$ is defined in \eqref{eq:gamma-bnd}, and fix $t \in \mathbb{R}$. Given are circuit descriptions of a unitary representation $\{U(g)\}_{g \in G}$ of a group~$G$ and a classical description of a $k$-local or sparse Hamiltonian~$H$. Decide which of the following holds:
\begin{align}
\textit{Yes:} & \quad \frac{1}{|G|} \sum_{g \in G} \left\Vert [U(g), e^{-iHt}] \right\Vert^2_\infty \geq \alpha,\\
\textit{No:} & \quad \frac{1}{|G|} \sum_{g \in G} \left\Vert [U(g), e^{-iHt}] \right\Vert^2_\infty \leq \beta,
\end{align}
\end{problem}

In what follows, we show that Ham-Sym-Avg-Spec is in QAM, and it remains an interesting open question to determine whether this problem is QAM-Hard or hard for some other complexity class.

\begin{theorem}
\label{thm:QAM-HamSym-test} The promise problem Ham-Sym-Avg-Spec is in QAM.
\end{theorem}

\begin{proof}
\begin{figure}
    \centering
    \includegraphics[width=1.0\columnwidth]{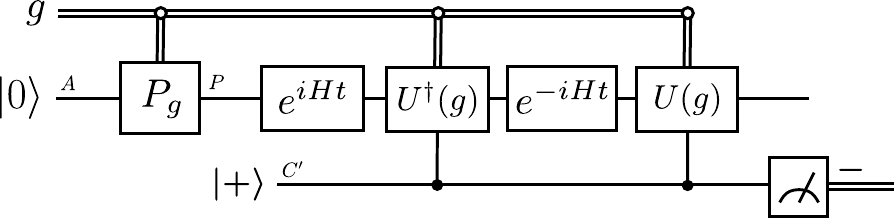}
    \caption{Circuit depicting a QAM test for Hamiltonian symmetry with respect to a group, where it is understood that the unitary $P$ is implemented by an all-powerful prover. The final measurement is in the Hadamard basis, and the algorithm accepts if the $\ket{-}$ outcome occurs.}
    \label{fig:ham-sym-avg-spec}
\end{figure}
Consider the following steps of a QAM interactive proof (see Figure~\ref{fig:ham-sym-avg-spec}):
\begin{enumerate}
    \item The verifier and prover are given a value $g \in G$ chosen uniformly at random.
    \item The prover prepares a state $\ket{\psi_g}$ in register $P$, which depends on the value $g$ and which has dimension equal to that of $H$.
    \item The verifier adjoins a qubit $C'$ in the state $|+\rangle$, performs the Hamiltonian evolution $e^{iHt}$, the controlled unitary $|0\rangle\!\langle 0| \otimes I + |1\rangle\!\langle 1| \otimes U^\dag(g) $, the Hamiltonian evolution $e^{-iHt}$, and the controlled unitary $|0\rangle\!\langle 0| \otimes I + |1\rangle\!\langle 1| \otimes U(g) $.
    \item The verifier measures the qubit $C'$ in the Hadamard basis $\{|+\rangle,|-\rangle\}$ and accepts if the outcome $|-\rangle$ occurs.
\end{enumerate}

Let us now analyze the acceptance probability of this interactive proof. We define the set of states $\{\ket{\psi_g} = P_g\ket{0}\}_{g \in G}$. 
Conditioned on the value $g$, the prover's state is $\ket{\psi_g}$, and defining the unitary $W(g,t) \equiv U(g)e^{-iHt} U^\dag(g)e^{iHt}$, the acceptance probability of Steps~3-4 is then given by
\begin{multline}
      \left\Vert (\bra{-}_{C'} \otimes I_P) \frac{1}{\sqrt{2}} (|0\rangle_{C'} |\psi_g\rangle_P +|1\rangle_{C'} W(g,t) |\psi_g\rangle_P) \right\Vert_2^2  \\
     = \frac{1}{4} \left\Vert  (I - W(g,t) )|\psi_g\rangle_P \right\Vert_2^2.
\end{multline}
Thus, for a fixed set $\{P_g\}_g$ of prover unitaries and averaging over the shared uniform randomness, the acceptance probability is given by
\begin{equation}
    \frac{1}{4|G|} \sum_{g\in G} \left\Vert  (I - W(g,t) )|\psi_g\rangle_P \right\Vert_2^2.
\end{equation}
Finally maximizing over all such prover unitaries leads to the following expression for the acceptance probability:
\begin{align}
    & \max_{\{P_g\}_g} \frac{1}{4|G|} \sum_{g\in G} \left\Vert  (I - W(g,t) )|\psi_g\rangle_P \right\Vert_2^2 \notag \\
    & = \frac{1}{4|G|} \sum_{g \in G} \left\Vert  I - W(g,t) \right\Vert_{\infty}^2 \notag \\
    & = \frac{1}{4|G|} \sum_{g \in G} \left\Vert [U(g), e^{-iHt}] \right\Vert^2_\infty,
\end{align}
where the reasoning is the same as that in \eqref{eq:QMA-block-eq}.
As the acceptance probability of this QAM interactive proof is precisely related to the decision criteria in Problem~\ref{prob:QAM-HamSym-test}, this concludes the proof.
\end{proof}

\section{Conclusion}

\label{sec:conclusion}

In summary, we established the computational complexity of various symmetry-testing problems. In particular, we showed that the various problems are complete for BQP, QMA, QSZK, QIP(2), \qipeb, and QIP, encompassing much of the known suite of quantum interactive proof models. In some cases, we employed the interactive proof algorithms from \cite{LRW22} to establish containment, and in other cases, we devised new algorithms. We proved hardness results by embedding various circuits involved in a given computation into the preparation of a state or channel or into a unitary representation of a group. Finally, we introduced two Hamiltonian symmetry-testing problems and proved that they are contained in QMA and QAM. 

Going forward from here, there are several directions to consider:
\begin{itemize}
    \item Let us observe that several key resources such as entanglement and distinguishability have been connected to the quantum interactive proof hierarchy, through the findings of \cite{HMW13,GHMW15} and \cite{kitaev2000parallelization,W02QIP,RW05,W09zkqa,HMW13,RASW23}, respectively. Our work makes a nontrivial link between this hierarchy and asymmetry, another key resource. These connections make us wonder whether other resources in quantum mechanics, such as coherence, magic, athermality, etc.~\cite{CG18}, can be linked with the same hierarchy.
    \item We are curious whether the two aforementioned Hamiltonian symmetry-testing problems could be shown to be complete for QMA and QAM, respectively, or complete for some other quantum complexity class of interest.
    \item Several multipartite separability problems were identified in \cite{philip2023quantum} and related to a quantum interactive proof setting in which there is a prover who performs a measurement, sends the classical outcome to multiple provers, who then send states to the verifier. One could thus try to find a symmetry-testing problem that is complete for this class.
    \item Various quantum algorithms for testing symmetries of channels, measurements, and Lindbladians under the Hilbert--Schmidt norm were proposed recently in \cite{bandyopadhyay2023efficient}. One could attempt to show that corresponding symmetry-testing problems are complete for BQP.
\end{itemize} 

\noindent \textbf{Data availability statement:} No data was generated for this paper.

\noindent \textbf{Competing interests:}  The authors declare there are no competing interests.

\begin{acknowledgments}
We acknowledge the guiding role that Professor A.~Ravi~P.~Rau has played in our academic lives, through many influential scientific discussions and interactions. We take this occasion to dedicate our  paper, in which symmetry has played an essential role, to Prof.~Rau. 

We thank Aby Philip and Vishal Singh for discussions.
MLL acknowledges support from the DoD SMART Scholarship program. She also acknowledges support from NSF Grant No.~2315398 for a visit to Cornell University during November 2022, when this work was initiated. 
MMW and SR acknowledge support from NSF Grant No.~2315398.
\end{acknowledgments}

\bibliography{Ref}

\appendix

\section{Error and Number of Samples in State-HS-Symmetry}

\label{app:error-mixed-state-HS-symmetry}

In Theorem~\ref{thm:BQP-Sym-HS}, we proved that the problem State-HS-Symmetry is BQP-Complete. In this section, we discuss the number of samples required to obtain the desired accuracy and confidence. To do this, we invoke Hoeffding's bound.

\begin{lemma}[Hoeffding's Bound \cite{H63}]
\label{lem:hoeffding}
Suppose that we are given $n$ independent samples $Y_1, \ldots, Y_n$ of a bounded random variable $Y$ taking values in the interval $[a,b]$ and having mean $\mu$. Set 
\begin{equation}
    \overline{Y_n} \coloneqq \frac{1}{n} (Y_1 + \cdots +Y_n)
\end{equation}
to be the sample mean. Let $\varepsilon > 0 $ be the desired accuracy, and let $1-\delta$ be the desired success probability, where $\delta \in (0,1)$. Then
\begin{equation}
\label{eq:Hoeff-bound}
\Pr[\vert \overline{Y_n} - \mu \vert \leq \varepsilon] \geq 1-\delta,
\end{equation}
if
\begin{equation}
    n \geq \frac{M^2}{2\varepsilon^2} \ln \!\left( \frac{2}{\delta}\right),
\end{equation}
where $M \coloneqq b -  a$.
\end{lemma}

In the main text, we mapped a general BQP algorithm to State-HS-Symmetry. In a general BQP algorithm, we measure a single qubit called the decision qubit, leading to a random variable $Y$ taking the value $0$ with probability $p_{\operatorname{rej}}$ and the value $1$ with probability $p_{\operatorname{acc}}$, where $p_{\operatorname{acc}}$ is the acceptance probability of the algorithm. We repeat this procedure $n$ times and label the outcomes $Y_1, \ldots, Y_n$. We output the mean
\begin{equation}
    \overline{Y_n} = \frac{1}{n}\left(Y_1 + \cdots + Y_n\right)
\end{equation}
as an estimate for the true value $p_{\operatorname{acc}}$
\begin{equation}
    p_{\operatorname{acc}} = \bra{x}_S \bra{0}_A Q^\dagger (\outerproj{1}_D \otimes I_G) Q \ket{x}_S \ket{0}_A.
\end{equation}
By plugging into Lemma~\ref{lem:hoeffding}, setting
\begin{equation}
\mu=p_{\operatorname{acc}}
\end{equation}
therein, and taking $n$ to satisfy the condition $n \geq \frac{1}{2\varepsilon^2} \ln\! \left(\frac{2}{\delta}\right)$, we can achieve an error $\varepsilon$ and confidence $\delta$ (as defined in \eqref{eq:Hoeff-bound}).

Now, we see from \eqref{eq:mixed-HS-sym-accep-prob} that the modified algorithm has an acceptance probability $1-p_{\operatorname{rej}}^2$, i.e., equal to one minus the square of the original BQP algorithm's rejection probability. In the modified algorithm, we measure the decision qubit, leading to a random variable $Z$ taking value $0$ with probability $p_{\operatorname{rej}}^2$ and the value $1$ with probability $1-p_{\operatorname{rej}}^2$. We repeat the procedure $m$ times and label the outcomes $Z_1, \ldots, Z_m$. We output the mean 
\begin{equation}
    \overline{Z}_m = 1 - \frac{1}{m} \left(Z_1 + \cdots + Z_m \right)
\end{equation}
as an estimate for the true value $p_{\operatorname{rej}}^2$. Setting $\tilde{\mu} = p_{\operatorname{rej}}^2$, and plugging into Lemma~\ref{lem:hoeffding}, it follows that
\begin{equation}
\Pr[\vert \overline{Z}_m - \tilde{\mu} \vert \leq \varepsilon^2] \geq 1-\delta,
\end{equation}
if
\begin{equation}
    m \geq \frac{1}{2\varepsilon^4} \ln \!\left( \frac{2}{\delta}\right).
\end{equation}
Consider the following inequalities:
\begin{align}
    \varepsilon^2 &\geq \left\vert \overline{Z}_m - \tilde{\mu} \right\vert \nonumber \\
    &= \left\vert \overline{Z}_m - \mu^2 \right\vert \nonumber \\
    &= \left\vert \sqrt{\overline{Z}_m} - \mu \right\vert \left\vert \sqrt{\overline{Z}_m} + \mu \right\vert \nonumber \\
    &\geq \left\vert \sqrt{\overline{Z}_m} - \mu \right\vert^2,
\end{align}
where the second inequality is derived from the fact that $\overline{Z}_m, \mu \in [0,1]$, so that $\left\vert\overline{Z}_m + \mu \right\vert \geq \left\vert \overline{Z}_m - \mu \right\vert$. Thus,
\begin{equation}
    \left\vert \sqrt{\overline{Z}_m} - \mu \right\vert \leq \varepsilon.
\end{equation}
In other words, 
\begin{equation}
    \varepsilon^2 \geq \left\vert \overline{Z}_m - \mu^2 \right\vert \implies \varepsilon \geq \left\vert \sqrt{\overline{Z}_m} - \mu \right\vert
    \end{equation}
    so that
    \begin{align}
    \Pr\left[\left\vert \sqrt{\overline{Z}_m} - \mu \right\vert \leq \varepsilon\right] &\geq \Pr[\left\vert \overline{Z}_m - \mu^2 \right\vert \leq \varepsilon^2] \nonumber \\
    &\geq 1-\delta.
\end{align}

Thus, $\sqrt{\overline{Z}_m}$ is an estimator for $p_{\operatorname{rej}}$ and taking 
\begin{equation}
    m \geq \frac{1}{2\varepsilon^4} \ln \!\left( \frac{2}{\delta}\right)
\end{equation}
suffices to achieve an error $\varepsilon$ and confidence $\delta$ in estimating $p_{\operatorname{rej}}$.

\end{document}